\documentclass[11pt,english,reqno]{article}
\usepackage{lmodern}
\usepackage{lmodern}

\usepackage[T1]{fontenc}
\usepackage[latin9]{inputenc}
\usepackage{float}
\usepackage{enumitem}
\usepackage{tipa}
\usepackage{tipx}
\usepackage{amsmath}
\usepackage{amsthm}
\usepackage{amssymb}
\usepackage{graphicx}
\usepackage{color}

\makeatletter
\theoremstyle{plain}
\newtheorem{thm}{\protect\theoremname}
  \theoremstyle{plain}
  \newtheorem{lem}[thm]{\protect\lemmaname}
  \newtheorem{coro}[thm]{\protect \corollaryname}

\newtheoremstyle{myremark} 
    {\topsep}                    
    {\topsep}                    
    {\rm}                        
    {}                           
    {\bf}                        
    {.}                          
    {.5em}                       
    {}  

\usepackage{fullpage}
\usepackage{relsize}

\allowdisplaybreaks

\newcommand{\Tau}{\mathrm{T}}
\newcommand{\Kappa}{\mathrm{K}}

\def\<{\langle}
\def\>{\rangle}
\def\|{\Vert}

\def\bw{{\boldsymbol w}}
\def\reals{{\mathbb R}}
\def\naturals{{\mathbb N}}

\def\sT{{\sf T}}
\def\sb{{\sf b}}
\def\sd{{\sf d}}
\def\hsb{\hat{\sf b}}
\def\hsd{\hat{\sf d}}
\def\Tr{{\rm Tr}}
\def\op{\mbox{\tiny\rm op}}

\def\ed{\stackrel{{\rm d}}{=}}
\def\semp{\mbox{\tiny emp}}
\def\id{{\boldsymbol I}}
\def\normal{{\sf N}}
\def\div{{\rm div}}
\def\bm{{\boldsymbol m}}
\def\bu{{\boldsymbol u}}
\def\bv{{\boldsymbol v}}
\def\hbu{\hat{\boldsymbol u}}
\def\hbv{\hat{\boldsymbol v}}
\def\ba{{\boldsymbol a}}
\def\bb{{\boldsymbol b}}
\def\bc{{\boldsymbol c}}
\def\bh{{\boldsymbol h}}
\def\bq{{\boldsymbol q}}
\def\bx{{\boldsymbol x}}
\def\by{{\boldsymbol y}}
\def\bz{{\boldsymbol z}}
\def\br{{\boldsymbol r}}
\def\bA{{\boldsymbol A}}
\def\bB{{\boldsymbol B}}
\def\bC{{\boldsymbol C}}
\def\bH{{\boldsymbol H}}
\def\bM{{\boldsymbol M}}
\def\bQ{{\boldsymbol Q}}
\def\bP{{\boldsymbol P}}
\def\bS{{\boldsymbol S}}
\def\bX{{\boldsymbol X}}
\def\bY{{\boldsymbol Y}}
\def\bZ{{\boldsymbol Z}}
\def\bG{{\boldsymbol G}}

\def\bU{{\boldsymbol U}}
\def\bV{{\boldsymbol V}}
\def\balpha{{\boldsymbol \alpha}}
\def\hbtheta{\hat{\boldsymbol \theta}}
\def\btheta{{\boldsymbol \theta}}
\def\bSigma{{\boldsymbol \Sigma}}
\def\bKappa{{\boldsymbol {\rm K}}}

\def\GOE{{\rm GOE}}
\def\NMSE{{\rm NMSE}}

\def\bfone{{\boldsymbol 1}}
\def\bzero{{\boldsymbol 0}}
\def\cA{{\mathcal A}}
\def\cH{{\mathcal H}}
\def\cC{{\mathcal C}}
\def\cK{{\mathcal K}}
\def\SVT{{\sf S}}
\def\sP{{\sf P}}
\def\sR{{\sf R}}
\def\vec{{\rm vec}}
\def\eps{\varepsilon}

\newcommand{\N}{\mathbb{N}}
\newcommand{\E}{\mathbb{E}}
\newcommand{\R}{\mathbb{R}}
\newcommand{\approxP}{\mathrel{\stackrel{{\rm P}}{\mathrel{\scalebox{1.8}[1]{$\simeq$}}}}}

\usepackage{babel}
\providecommand{\corollaryname}{Corollary}
  \providecommand{\lemmaname}{Lemma}
  
\providecommand{\theoremname}{Theorem}

\usepackage{bbm}

\@ifundefined{showcaptionsetup}{}{%
 \PassOptionsToPackage{caption=false}{subfig}}
\usepackage{subfig}
\makeatother

\usepackage{babel}
  \providecommand{\lemmaname}{Lemma}
\providecommand{\theoremname}{Theorem}

\theoremstyle{myremark}
\newtheorem{remark}{Remark}[section]

\begin{document}

\title{State Evolution for Approximate Message Passing\\ with Non-Separable Functions}

\author{Rapha\"el Berthier\thanks{Ecole Normale Supérieure, Paris and Université Paris-Sud, Orsay}, \;\;
Andrea  Montanari\thanks{Department of Electrical Engineering and Department
  of Statistics, Stanford University}, \;\; and\;\;
Phan-Minh Nguyen\thanks{Department of Electrical Engineering, Stanford University}}

\maketitle

\begin{abstract}
Given a high-dimensional data matrix $\bA\in\reals^{m\times n}$, Approximate Message Passing (AMP)  algorithms construct sequences
of vectors $\bu^t\in\reals^n$, $\bv^t\in\reals^m$, indexed by $t\in\{0,1,2\dots\}$ by iteratively applying $\bA$ or $\bA^{\sT}$, and suitable non-linear functions,
which depend on the specific application. 
Special instances of this approach have been developed --among other applications-- for compressed sensing reconstruction, robust regression,
Bayesian estimation, low-rank matrix recovery, phase retrieval, and community detection in graphs. For certain classes of random matrices $\bA$,
AMP admits an asymptotically exact description in the high-dimensional limit $m,n\to\infty$, which goes under the name of \emph{state evolution.}

Earlier work established state evolution for separable non-linearities (under certain regularity conditions). Nevertheless,
empirical work demonstrated several important applications that require non-separable functions.
In this paper we generalize state evolution to Lipschitz continuous non-separable nonlinearities, for Gaussian matrices $\bA$. Our proof
makes use of Bolthausen's conditioning technique along with several approximation arguments. In particular, we introduce a modified 
algorithm (called LAMP for Long AMP) which is of independent interest.
\end{abstract}

\section{Introduction}

Over the last few years Approximate Message Passing (AMP) algorithms have been applied to a broad range of statistical estimation problems,
including compressed sensing \cite{donoho2009message}, robust regression \cite{donoho2016high}, Bayesian estimation \cite{kamilov2012approximate}, low rank matrix
recovery \cite{kabashima2016phase}, phase retrieval \cite{schniter2015compressive}, and  community detection in graphs \cite{deshpande2016asymptotic}.
In a fairly generic formulation\footnote{More general settings have
  also been developed, see for instance \cite{Javanmard2013}.}, AMP
takes as input a random data matrix $\bA\in \reals^{m\times n}$ and generates sequences of vectors $\bu^t\in\reals^n$, $\bv^t\in\reals^m$,
indexed by $t\in \naturals$ according to the iteration
\begin{align}
\bu^{t+1} & =\bA^{\sT}g_{t}(\bv^{t})-\sd_{t}e_{t}(\bu^{t})\, ,\label{eq:FIRSTasymmetricAMP_1}\\
\bv^{t} & =\bA e_{t}(\bu^{t})-\sb_{t}g_{t-1}(\bv^{t-1})\, .\label{eq:FIRSTasymmetricAMP_2}
\end{align}
Here $g_t:\reals^m\to\reals^m$ and $e_t:\reals^n\to\reals^n$ are two sequences of functions indexed by the iteration number $t$, 
that encode the specific application. The coefficients $\sd_t$, $\sb_t\in\reals$ are completely fixed by the choice of these functions.  
For instance, assuming $\E\{A_{ij}^2\}= 1/m$, we can use
\begin{equation}
  \sd_{t}^{\semp}=\frac{1}{m}{\rm div}\,g_{t}(\bv^t)\, ,\quad\sb^{\semp}_{t}=\frac{1}{m}{\rm div}\, e_{t}(\bu^{t})\, . \label{eq:FIRSTonsager_asym}
\end{equation}
(A slightly different definition, that is more convenient for proofs, will be adopted in Section \ref{sec:MainResults}.)

Apart from being broadly applicable, AMP algorithms admit an asymptotically exact characterization in the high-dimensional limit 
$m,n\to \infty$ with $m/n$ converging to a limit, which is known as \emph{state evolution}. Informally, for any $t$ fixed,  in the high-dimensional limit,
$\bu^t$ is approximately Gaussian with mean zero and covariance $\tau^2_t\id_n$, while 
$\bv^t$ is approximately $\normal(0,\sigma^2_t\id_m)$. The variance parameters $\tau^2_t, \sigma^2_t$ can be computed via a one-dimensional recursion.

State evolution was proved  in \cite{Bayati2011}  for the recursion (\ref{eq:FIRSTasymmetricAMP_1}), (\ref{eq:FIRSTasymmetricAMP_2})
under two key assumptions
\begin{itemize}
\item  $\bA$ a Gaussian random matrix with with i.i.d. entries $(A_{ij})_{i\le m,j\le n}\sim\normal(0,1/m)$.
\item The functions $g_t(\,\cdot\,)$, $e_t(\,\cdot\,)$ are separable\footnote{We say that $f:\reals^d\to\reals^d$ is separable if 
$f(x_1,\dots,x_d)_i= f_i(x_i)$ for some functions $f_i:\reals\to\reals$.} and Lipschitz continuous.
\end{itemize}
This paper relaxes the second assumption and establishes state evolution for functions $g_t(\,\cdot\,)$, $e_t(\,\cdot\,)$
that are Lipschitz continuous but not necessarily separable. Our proof uses (as the original paper \cite{Bayati2011}) a conditioning
technique initially developed by Erwin Bolthausen \cite{bolthausen2014iterative} to study the TAP equations in spin glass theory.
A key difficulty with non-separable denoisers is that the iterates $g_1(\bv^1)$, $g_2(\bv^2)$, \dots, $g_t(\bv^{t})\in\reals^m$ might be collinear and lie in a subspace of
dimension smaller than $t$, for large $m$. This degeneracy (or a similar problem with the $e_1(\bu^1)$, $e_2(\bu^2)$, \dots $e_t(\bu^{t})$) would cause
a naive adaptation of the proof of \cite{Bayati2011} to break down. In order to circumvent this problem without introducing ad hoc assumptions, 
we proceed in two steps:
\begin{enumerate}
\item We introduce a random perturbation of the functions $e_t(\,\cdot\,)$, $g_t(\,\cdot\,)$. We prove that, with probability one with respect to this random perturbation,
the new iteration satisfies the required non-degeneracy assumption.
\item We prove that both AMP and state evolution are uniformly continuous in the size of the perturbation, and hence we can let the perturbation
vanish recovering state evolution for the original unperturbed problem.
\end{enumerate}
Further, we obtain a streamlined proof with respect to the strategy of \cite{Bayati2011}, by introducing a different algorithms,
that we call LAMP (for Long AMP). State evolution is proved first for LAMP, and then the latter is shown to be closely approximated by the original AMP.
We believe that LAMP is potentially of independent interest and will be further investigated in  \cite{OurUnpub}

In the rest of this introduction we will briefly describe two applications of AMP with non-separable 
nonlinearities, and show how state evolution can be used to characterize its behavior. Both of these are examples of generalized
compressed sensing, cf. Section \ref{sec:Application-CS}.
We will then review some related work in Section \ref{sec:Related}, and state our results in Section \ref{sec:MainResults} (for the asymmetric iteration (\ref{eq:FIRSTasymmetricAMP_1})
and Section \ref{sec:Symmetric-AMP} (for the analogue case in which $\bA$ is a random symmetric matrix)).
Proofs are presented in Sections \ref{sec:ProofSymmetric} and \ref{sec:Asymmetric-AMP}. In fact, we will first 
prove state evolution in the case in which $\bA$ is a symmetric random matrix, and then reduce the asymmetric case to the symmetric one.
Finally, Section \ref{sec:Application-CS} applies the general theory to compressed sensing reconstruction with a variety of denoisers. In particular, we derive
a bound on the convergence rate for denoisers that are projectors onto convex sets.
Several technical elements are deferred to the appendices. 

For a summary of notations used throughout the paper, the reader is urged to consult Section \ref{sec:Notations}.

\subsection{Vignette $\# 1$: Matrix compressed sensing}

We want to reconstruct an unknown matrix $\bX_0\in\reals^{n_1\times n_2}$ from linear measurements $\by\in\reals^m$, where
\begin{align}
\by = \cA(\bX_0)\, .\label{eq:MatrixSensing}
\end{align}
Here $\cA:\reals^{n_1\times n_2}\to\reals^m$ is a Gaussian linear operator. Concretely $y_i = \<\bA_i,\bX_0\>$ where $\bA_i\in\reals^{n_1\times n_2}$ 
are i.i.d. matrices with independent entries $(\bA_i)_{r,c}\sim\normal(0,1/m)$. This setting was first studied in \cite{recht2010guaranteed} and
can be used as a simple model for system identification and matrix completion. 

The following AMP algorithm can be used to reconstruct $\bX_0$ from observations $\by$:
\begin{align}
\bX^{t+1} & =\SVT\big(\bX^{t}+\cA^{\sT}\br^{t}; \lambda_t\big)\, ,\label{eq:MatrixAMP1}\\
\br^{t} & =\by-\cA(\bX^{t})+\sb_t \br^{t-1}\, ,\label{eq:MatrixAMP2}
\end{align}
with initialization $\bX^0 = \bzero$. After $t$ iterations, the algorithm produces an estimate $\bX^t$, and a residual $\br^t$.
Here $\cA^{\sT}$ is the adjoint\footnote{We can represent the action $\cA(\bX)$ by vectorizing $\bX$ as $\vec(\bX)\in\reals^n$,
$n=n_1n_2$. If $\bA\in\reals^{m\times n}$ is the matrix whose $i$-th row is $\bA_i$, then $\cA(\bX) = \bA\vec(\bX)$. Then the adjoint
$\cA^{\sT}$ corresponds to the transpose $\bA^{\sT}$.} of the operator $\cA$, 
\begin{align}
\sb_t = \frac{1}{m}\div\,\SVT \big(\bX^{t-1}+\cA^{\sT}\br^{t-1}; \lambda_{t-1}\big)\, ,\label{eq:OnsagerSVT}
\end{align}
and $\SVT(\,\cdot\,;\lambda)$ is the singular value thresholding (SVT) operator, defined as follows.
For a matrix $\bY\in\reals^{n_1\times n_2}$, with singular value decomposition
\begin{align}
\bY=\sum_{i=1}^{n_{1}\wedge n_{2}}\sigma_{i}\bu_{i}\bv_{i}^{\sT},
\end{align}
the SVT operator yields
\begin{equation}
\SVT(\bY;\lambda)=\sum_{i=1}^{n_{1}\wedge n_{2}}(\sigma_{i}-\lambda)_{+}\bu_{i}\bv_{i}^{\sT}\, .
\end{equation}
The divergence in Eq.~(\ref{eq:OnsagerSVT}) can be computed explicitly using a formula from \cite{Candes,donoho2014minimax}, see Appendix \ref{app:Matrix}.
The sequence of parameters $(\lambda_t)_{t\ge 0}$ can be chosen to optimize the algorithm performance. 

Fixed points of this AMP algorithm are minimum nuclear norm solution of the constraint $\by=\bA(\bX)$. This
algorithm was implemented in \cite{donoho2013unpub} and partly motivated the predictions of \cite{donoho2013phase}.
A recent detailed study (and generalizations) can be found in \cite{Gavish2017}, showing that its phase transition matches the one of nuclear 
norm minimization, predicted in \cite{donoho2013phase} and proved in \cite{oymak2013squared,amelunxen2014living}.

With a change of variables, the algorithm (\ref{eq:MatrixAMP1}),  (\ref{eq:MatrixAMP2}) can be recast in the general form (\ref{eq:FIRSTasymmetricAMP_1}),
(\ref{eq:FIRSTasymmetricAMP_2}) with one of the functions being non-separable and given by the SVT operator (the change of variables is described in Section \ref{sec:Application-CS}). 

\begin{figure}
\includegraphics[width=0.5\textwidth]{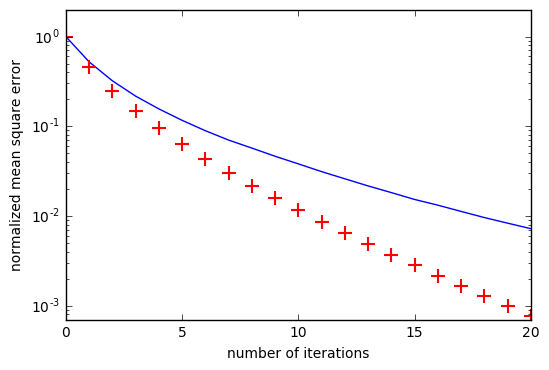} \includegraphics[width=0.5\textwidth]{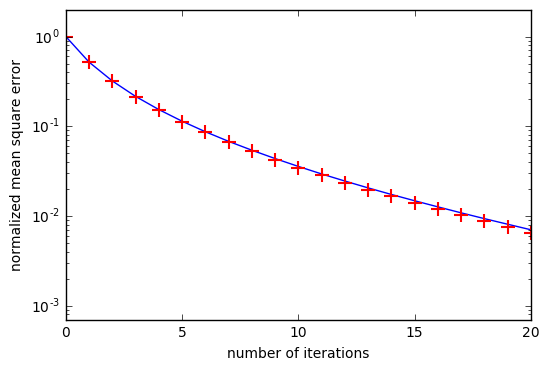}
\caption{Matrix compressed sensing reconstruction using AMP: normalized mean square error as a function of the number of iterations.
Left : $30\times 30$ matrices of rank $3$. Right: $170\times 170$
matrices of rank $17$. Red pluses ($+$): simulations. Blue line: state evolution prediction.}
\label{fig:matrix_reconstruction_se} 
\end{figure}

In Figure \ref{fig:matrix_reconstruction_se} we report the results of numerical simulations using this algorithm. We generated
$\bX_0\in\reals^{n_1\times n_2}$ of rank $r$ by letting $\bX_0 = \bU\bV^{\sT}$ for $\bU\in\reals^{n_1\times r}$, $\bV\in\reals^{n_2\times r}$
uniformly random orthogonal matrices, and computed $m$ measurements $\by$ as per Eq.~(\ref{eq:MatrixSensing}). We took $n_1=n_2$, $r =0.1\cdot n_1$,
$m = 0.65 \cdot n_1n_2$.  We chose the threshold parameter $\lambda_{t}$ to be proportional to
the noise level as estimated via the residual \cite{montanari2012graphical}:
\begin{equation}
\lambda_{t}=2\sqrt{n_{1}} \, \frac{\| \br^{t}\|_{2}} {\sqrt{m}}.
\end{equation}
We plot the the normalized mean square error as a function of the iteration number (with $n=n_1n_2$ the number of unknowns):
\begin{align}
\NMSE(t;n)= \frac{\| \bX^{t}-\bX_0\|_{F}^{2}}{\|\bX_0\|_{F}^{2}}\, .
\end{align}
State evolution allows to predict the value $\lim_{n\to\infty} \NMSE(t;n)$. The prediction is already very accurate for $n_1=n_2=170$.

\subsection{Vignette $\# 2$: Compressed sensing with images}
\label{sec:NLM-AMP}

We represent an image as a two-dimensional array $\bx = (x_{i,j})_{i\le n_1,j\le n_2}$, which we identify with its vectorization $\vec(\bx)\in\reals^{n}$, $n=n_1n_2$. In compressed sensing we acquire a small 
number of incoherent measurements $\by\in\reals^m$ according to 
\begin{align}
\by = \bA\bx +\bw\, .
\end{align}
where $\bA\in\reals^{m\times n}$ is a known sensing matrix for which we assume the simple Gaussian model $(A_{ij})_{i\le m,j\le n}\sim_{iid}\normal(0,1/m)$,
and $\bw\sim\normal(0,\sigma^2_w\id_m)$ is noise. 

A broad class of AMP reconstruction algorithms take the form
\begin{align}
\bx^{t+1} & =\eta_{t}\big(\bx^{t}+\bA^{\sT}\br^{t}\big)\, ,\label{eq:AMPImage1}\\
\br^{t} & =\by-\bA\bx^{t}+\sb_t\br^{t-1}\, .\label{eq:AMPImage2}
\end{align} 
where $\bx^0=0$, $\eta_t:\reals^{n_1\times n_2}\to \reals^{n_1\times n_2}$ is a sequence of image denoisers, and 
\begin{align}
\sb_t = \frac{1}{m}\div\,\eta_{t-1}\big(\bx^{t-1}+\bA^{\sT}\br^{t-1}\big)\, .
\end{align}
The compressed sensing reconstruction algorithm in \cite{donoho2009message} was a special case of this iteration with $\eta_t(\,\cdot\,)$ 
corresponding to coordinate-wise soft thresholding (in a suitable basis), hence leading to a separable AMP. Several authors studied the same algorithm
with non-separable denoisers, including Hidden Markov Models \cite{schniter2010turbo,som2012compressive}, total variation and block thresholding denoisers \cite{donoho2013accurate},
universal denoising \cite{ma2016approximate},
restricted Boltzmann machines \cite{tramel2016approximate}. As documented in these papers, a good choice of the denoiser yields a significant performance boost
over classical compressed sensing reconstruction methods, such as $\ell_1$ minimization.

Again, the iteration (\ref{eq:AMPImage2}), (\ref{eq:AMPImage1}) can be put in the form (\ref{eq:FIRSTasymmetricAMP_1}), (\ref{eq:FIRSTasymmetricAMP_2}) with a change of variables described
in Section \ref{sec:Application-CS}. A non-separable denoiser $\eta_t$ translates into non-separable non-linearities $g_t$, $e_t$.

Here we use Non-Local Means denoising (NLM) \cite{buades:hal-00271141}. Given a noisy image $\bz$, NLM estimates pixel $(i,j)$
as a weighted average of the pixels of $\bz$: 
\begin{align}
\eta(\bz)_{i,j}=\frac{ \sum_{(k,l)} W_{(k,l),(i,j)} (\bz)\, z_{k,l} }{ \sum_{(k,l)} W_{(k,l),(i,j)} (\bz) }\, .\label{eq:NLM_average}
\end{align}
The weights $W_{(k,l),(i,j)}(\bz)$ depend on the similarity between  the patches in $\bz$ centered around $(k,l)$
and $(i,j)$ respectively, as well as on the distance between the two pixels. In a simple instantiation, we choose a patch size
$L\in\N_{>0}$, a range $R>0$,  and a precision parameter $h>0$. For a position $(k,l)$ in
the image, denote by $P_{(k,l)}(\bx)$ the subimage of $\bx$ (or patch) centered in $(k,l)$, of size $L\times L$. Then: 
\begin{align}
W_{(k,l),(i,j)}(\bz)=\bfone_{\|(i,j)-(k,l)\|\le R} \, \exp\left(-\frac{\left\Vert P_{(k,l)}(x)-P_{(i,j)}(x)\right\Vert _{2}^{2}}{L^{2}h^{2}}\right).\label{eq:NLM_weights}
\end{align}
In words, NLM averages patches that are similar to each other. The recent paper \cite{Metzler2016} studies this algorithm and demonstrates 
state-of-the-art performances. 
Here we carry out similar simulations to demonstrate the accuracy of the state evolution prediction. At each iteration we can choose three parameters:
$L_t$, $R_t$ and $h_t$. We fix $L_t = 7$, $R_t=11$ and adapt $h_t$ to the noise level. The theory developed in the next sections
suggests that $\|\br^t\|_2/\sqrt{m}$ is a good measure of the effective noise level after $t$ iterations. We therefore set
\begin{align}
h_{t}=0.9\cdot \frac{\| \br^{t}\|_{2}}{\sqrt{m}}\, ,
\end{align}
where the coefficient $0.9$ was selected empirically.

One difficulty is to compute the divergence of NLM denoisers $\div\,\eta_{t}$. Rather than computing explicitly the divergence from Eqs.~(\ref{eq:NLM_average}) and (\ref{eq:NLM_weights}),
we use a trick suggested in \cite[Section V.B]{Metzler2016}. The trick is based on  the formula 
\begin{align}
{\rm div}\,\eta_{t}(x)=\lim_{\eps\to 0}\E\left[\left\< \bZ,\frac{1}{\eps}\left(\eta_{t}(\bx+\eps \bZ)-\eta_{t}(\bx)\right)\right\> \right],\quad \bZ\sim\normal(0,\id_{n}).
\end{align}
Rather than taking the limit, we fix $\eps$ very small and evaluate the expectation by Monte Carlo. In high dimensions, concentration of
measure helps and it is sufficient to use only one or a few samples to approximate the integral.

\begin{figure}
\center \includegraphics[scale=0.27]{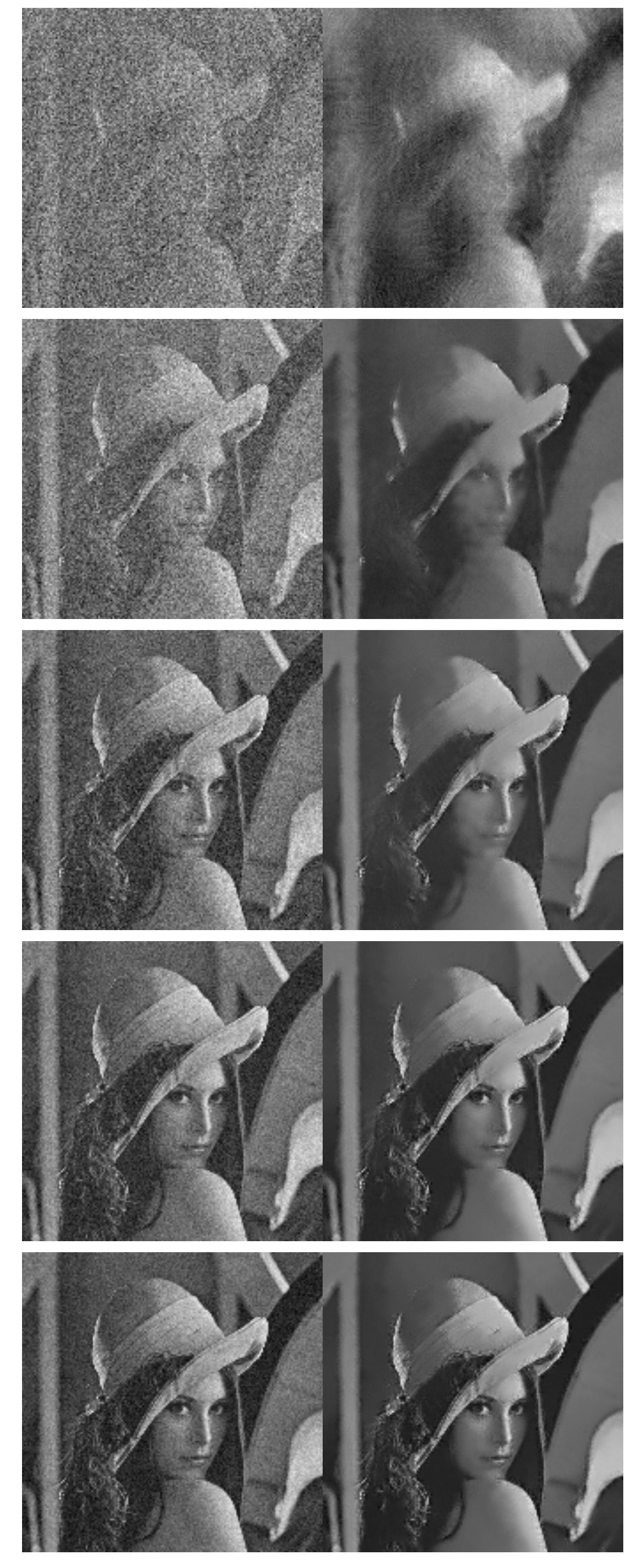} \caption{Compressed sensing reconstruction of Lena using NLM-AMP,
at undersampling ratio $m/n=0.5$: iterates $\bx^{t}+\bA^{\sT}\br^{t}$ (left column) and $\bx^{t+1}=\eta_{t}(\bx^{t}+\bA^{\sT}\br^{t})$
(right column) for  $t\in \{0,1,2,3,4\}$. (Details in the main text.)}
\label{fig:Lena_reconstruction} 
\end{figure}

\begin{figure}
\center \includegraphics[scale=0.7]{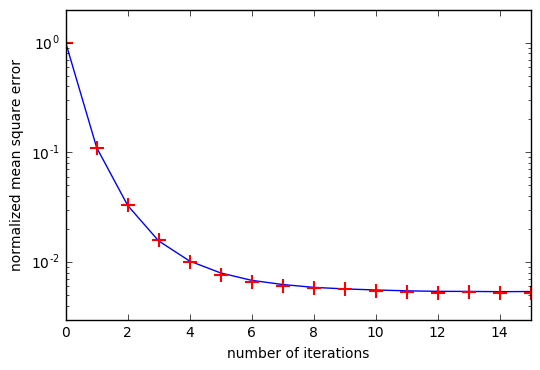} \caption{Compressed sensing reconstruction of Lena using NLM-AMP. Red pluses ($+$): evolution of the normalized square error. 
Blue line: state evolution prediction.}
\label{fig:Lena_reconstruction_se} 
\end{figure}

In Figure \ref{fig:Lena_reconstruction}, we demonstrate the algorithm performance 
for an image of size $170\times 170$ (i.e. $n_1=n_2=170$)  with  $m = 0.5\cdot n_{1}n_{2}$
measurements and noise level $\sigma_w = 0.034\cdot\| \bx_0 \|_2 / \sqrt{170}$. For each iteration $t\in \{0,1,2,3,4\}$, we show the estimates
$\bx^{t}+\bA^{\sT}\br^{t}$ (left column) together with the denoised versions $\bx^{t+1}=\eta_{t}(\bx^{t}+\bA^{\sT}\br^{t})$ (right column).
In Figure \ref{fig:Lena_reconstruction_se}, we  report the evolution
of the normalized square error $\NMSE(t;n) = \| \bx^{t}-\bx_0\|_{2}^{2}/\| \bx_0\|_{2}^{2}$,
as a function of the number of iteration. State evolution appears to track very closely the simulation results.

\section{Further related work}
\label{sec:Related}

Approximate Message Passing algorithms are motivated by ideas in spin glass theory, where they correspond to an iterative version of the celebrated TAP
equations \cite{thouless1977solution,bolthausen2014iterative}. They can also be derived from graphical models ideas, by viewing them as approximations of 
belief propagation \cite{koller2009probabilistic,montanari2012graphical}. In both of these cases, the AMP nonlinearities turn out to be related to  conditional expectation with respect to
certain prior distributions.  The theorems proved here apply more broadly,
as demonstrated by the example in Section \ref{sec:NLM-AMP}.

The state evolution analysis of \cite{Bayati2011} was generalized in a number of directions over the last few years. 
State evolution was proven to hold for matrices $\bA$ with i.i.d. subgaussian entries in \cite{bayati2015universality},
under the assumption that the non-linearity is a separable polynomial. The proof of \cite{bayati2015universality} is based on the moment method,
and hence is entirely different from the one presented here. 
Several generalizations of the basic iteration (\ref{eq:FIRSTasymmetricAMP_1}), (\ref{eq:FIRSTasymmetricAMP_2}) were studied in \cite{Javanmard2013}.
The framework of  \cite{Javanmard2013} allows to treat  some classes of matrices with independent Gaussian but not identically distributed entries,
as well as algorithms in which $\bu^t\in\reals^{n\times k}$, $\bv^t\in\reals^{m\times k}$ are matrices with $k$ fixed as $m,n\to\infty$.

A generalization of AMP to right-invariant random matrices was introduced and analyzed in 
\cite{ma2017orthogonal,rangan2016vector}, using the conditioning technique also applied here. This allows to treat classes of matrices
 with dependent entries and potentially large condition numbers. In the same direction, \cite{opper2016theory,ccakmak2017dynamical} develops 
iterative algorithms analogous to (\ref{eq:FIRSTasymmetricAMP_1}), (\ref{eq:FIRSTasymmetricAMP_2}) for unitarily invariant symmetric matrices,
and for compressed sensing. The analysis in these works is based on non-rigorous density functional methods from statistical physics.

All results discussed above are asymptotic, and characterize the limit $m,n\to\infty$ with $m/n$ converging to a limit. 
Nevertheless, the conditioning technique does rely on central-limit-theorem and concentration-of-measure arguments and,
as demonstrated in \cite{rush2016finite}, it can be sharpened to obtain non-asymptotic results.

Finally, a recent paper by Ma, Rush and Baron \cite{ma2017analysis} states a theorem establishing state evolution for 
compressed sensing reconstruction via AMP with a non-separable sliding-window  denoiser. The result of \cite{ma2017analysis} is
not directly comparable with ours, since it concerns a special class of non-separable nonlinearities, but provides non-asymptotic guarantees.

\section{Main results}
\label{sec:MainResults}

In this section we state our main result for the asymmetric AMP iteration of Eqs.~(\ref{eq:FIRSTasymmetricAMP_1}), (\ref{eq:FIRSTasymmetricAMP_2}).
A similar result for symmetric AMP will be stated in Section \ref{sec:Symmetric-AMP} (and proven in \ref{sec:ProofSymmetric}).

\subsection{Definitions}
For two sequences (in $n$) of random variables $X_{n}$ and $Y_{n}$, we write $X_{n}\approxP Y_{n}$ when their difference converges in probability to $0$, i.e. 
$X_{n}-Y_{n}\stackrel{{\rm P}}{\longrightarrow}0$.

For $\Kappa=\left(\Kappa_{s,r}\right)_{1\leq s,r\leq t}$ a $t\times t$
covariance matrix, we will write $(\bZ^{1},\dots,\bZ^{t})\sim\normal(0,\Kappa\otimes \id_{n})$
to mean that $\bZ^{1},\dots,\bZ^{t}$ is a collection of centered, jointly  Gaussian random vectors
in $\R^{n}$, with covariances $\E[\bZ^{s}(\bZ^{r})^{\sT}]=\Kappa_{s,r}\id_{n}$
for $1\leq s,r\leq t$.

For $k\in\mathbb{N}_{>0}$ and any $n,m\in\N_{>0}$, a function $\phi:\R^{n}\to\R^{m}$
is called \textit{pseudo-Lipschitz of order $k$} if there exists a constant $L$ such that for any $\bx,\by\in\R^{n}$,
\begin{equation}
\frac{\left\Vert \phi\left(\bx\right)-\phi\left(\by\right)\right\Vert _{2}}{\sqrt{m}}\leq L\left(1+\left(\frac{\left\Vert \bx\right\Vert _{2}}{\sqrt{n}}\right)^{k-1}+\left(\frac{\left\Vert \by\right\Vert _{2}}{\sqrt{n}}\right)^{k-1}\right)\frac{\left\Vert \bx-\by\right\Vert _{2}}{\sqrt{n}}.
\end{equation}
$L$ is then called the pseudo-Lipschitz constant of $\phi$. Note that this definition is the same as  introduced in \cite{Bayati2011}, apart from
a different scaling of the norm $\|\,\cdot\,\|_2$.  The normalization factors are introduced 
to simplify the analysis that follows.  For $k=1$, this definition coincides with the standard definition of a Lipschitz
function, for mapping between the normed spaces $(\reals^n,\|\,\cdot\,\|_2/\sqrt{n})$ and $(\reals^m,\|\,\cdot\,\|_2/\sqrt{m})$. 
In this case $L$ is the Lipschitz constant of $\phi$.

A sequence (in $n$) of pseudo-Lipschitz functions $\left\{ \phi_{n}\right\} _{n\in\mathbb{N}_{>0}}$
is called \emph{uniformly} pseudo-Lipschitz of order $k$ if, denoting by  $L_n$ is the pseudo-Lipschitz constant of order
$k$ of $\phi_n$, we have $L_n<\infty$ for each $n$ and $\limsup_{n\to\infty}L_{n}<\infty$. Note that the input
and output dimensions of each $\phi_{n}$ can depend on $n$. We call
any $L>\limsup_{n\to\infty}L_{n}$ a pseudo-Lipschitz constant of
the sequence.

\subsection{State evolution}

Fix $\delta>0$ and consider a sequence $m=m(n)\in\N$ such that
$m/n\to\delta$ as $n\to\infty$. For all $n$, we are given two sequences
of (deterministic) functions $\left\{ e_{t}:\mathbb{R}^{n}\to\mathbb{R}^{n}\right\} _{t\in\mathbb{N}}$
and $\left\{ g_{t}:\mathbb{R}^{m}\to\mathbb{R}^{m}\right\} _{t\in\mathbb{N}}$,
as well as a sequence of (deterministic) vectors $\bu^{0}=\bu^{0}\left(n\right)\in\mathbb{R}^{n}$,
and a sequence of random rectangular matrices $\bA=\bA\left(n\right)\in\mathbb{R}^{m\times n}$.

We next list  our assumptions (we refer to Section \ref{sec:Notations} for a summary of notations used in the paper):
\begin{enumerate}[font={\bfseries},label={(B\arabic*)}]
\item \label{assumption:gaussian-asym}$\bA$ has entries $(A_{ij})_{i\le m, j\le n} \sim_{iid}\normal(0,1/m)$.
\item For each $t\in\N$, the functions $e_{t}:\R^{n}\to\R^{n}$, $g_{t}:\R^{m}\to\R^{m}$
are uniformly Lipschitz (where uniformly is understood with respect to $n$). 
\item $\Vert \bu^{0}\Vert_{2}/\sqrt{n}$ converges to a finite constant as
$n\to\infty$. 
\item \label{assumption:u^0_conv}The following limit exists and is finite:
\begin{equation}
\Sigma_{0,0} \equiv \lim_{n\to\infty}\frac{1}{m}\big\< e_{0}(\bu^{0}),e_{0}(\bu^{0})\big\> .
\end{equation}
\item \label{assumption:mixed-conv-asym}For any $t\in\N_{>0}$ and any $s\ge 0$,
the following limit exists and is finite:
\begin{equation}
\lim_{n\to\infty}\frac{1}{m}\E\Big[\big\< e_{0}(\bu^{0}),e_{t}(\bZ)\big\> \Big]
\end{equation}
where $Z\sim\normal(0,s\id_{n})$.
\item \label{assumption:gaussian-conv-asym} For any $s,t\in\N_{>0}$ and any $\bS\in \reals^{2\times 2}$, $\bS\succeq \bzero$, 
the following limits exist and are finite: 
\begin{align}
\lim_{n\to\infty}\frac{1}{m}\E\Big[\big\< e_{s}(\bZ_{1}),e_{t}(\bZ_{2})\big\> \Big],\\
\lim_{n\to\infty}\frac{1}{m}\E\Big[\big\< g_{s}(\bZ_{3}),g_{t}(\bZ_{4})\big\> \Big],
\end{align}
where $(\bZ_{1},\bZ_{2})\sim\normal(0,\bS\otimes \id_{n})$
and $(\bZ_{3},\bZ_{4})\sim\normal(0,\bS\otimes \id_{m})$. 
\end{enumerate}
The technical assumptions \ref{assumption:u^0_conv},\ref{assumption:mixed-conv-asym}
and \ref{assumption:gaussian-conv-asym} allows to define two doubly infinite arrays
$(\Sigma_{s,r})_{s,r\ge 0}$ and $(\Tau_{s,r})_{ s,r\ge 1}$, through the following recursion, known as \emph{state evolution}.

We set $\Sigma_{0,0}$ using Assumption \ref{assumption:u^0_conv}. For each $t\ge 0$, given
$\left(\Sigma_{s,r}\right)_{0\leq s,r\leq (t-1)}$ and $\left(\Tau_{s,r}\right)_{1\leq s,r\leq t}$, we let, for $0\leq s\leq t$
\begin{align}
\Sigma_{t,s} & =\lim_{n\to\infty}\frac{1}{m}\E\left[\left\< e_{s}(\bZ_{\tau}^{s}),e_{t}(\bZ_{\tau}^{t}) \right\> \right]\, ,\\
\Tau_{t+1,s+1}& =\lim_{n\to\infty}\frac{1}{m}\E\left[\left\< g_{s}(\bZ_{\sigma}^{s}),g_{t}(\bZ_{\sigma}^{t})\right\> \right]\, ,
\end{align}
along with $\Tau_{s+1,t+1}=\Tau_{t+1,s+1}$ and $\Sigma_{s,t} = \Sigma_{t,s}$.
Here expectation is with respect to  $(\bZ_{\sigma}^{0},\dots,\bZ_{\sigma}^{t})\sim\normal (0,(\Sigma_{s,r})_{0\leq s,r\leq t}\otimes \id_{m})$,
$(\bZ_{\tau}^{1},\dots,\bZ_{\tau}^{t})\sim\normal(0,(\Tau_{s,r})_{1\leq s,r\leq t}\otimes \id_{n})$,
and it is understood that $\bZ_{\tau}^0 = \bu^0$. 

We will refer to the arrays $(\Sigma_{s,r})_{s,r\ge 0}$ and $(\Tau_{s,r})_{ s,r\ge 1}$ as to the \emph{state evolution
iterates} (and sometimes simply \emph{state evolution})  and denote them by $\left\{ \Tau_{s,t},\Sigma_{s,t}\middle\vert e_{t},g_{t},\bu^{0}\right\}$,
to make explicit the nonlinearities and initialization.

State evolution characterizes the AMP iteration of Eqs.~(\ref{eq:FIRSTasymmetricAMP_1}), (\ref{eq:FIRSTasymmetricAMP_2}), which we copy here for the reader's
convenience:
\begin{align}
\bu^{t+1} & =\bA^{\sT}g_{t}(\bv^{t})-\sd_{t}e_{t}(\bu^{t}),\label{eq:asymmetricAMP_1}\\
\bv^{t} & =\bA e_{t}(\bu^{t})-\sb_{t}g_{t-1}(\bv^{t-1})\, ,\label{eq:asymmetricAMP_2}
\end{align}
where the initial condition is given by $\bu^0$, and we let $g_{-1}(\,\cdot\,)=0$ by convention.
Further we use the following expression for the memory terms (which we shall refer to as `Onsager terms,' following the physics 
tradition):
\begin{equation}
\sd_{t}=\frac{1}{m}\E\left[\div\, g_{t}\left(\bZ_{\sigma}^{t}\right)\right],\quad\sb_{t}=\frac{1}{m}\E\left[\div\, e_{t}\left(\bZ_{\tau}^{t}\right)\right]\label{eq:onsager_asym}\, ,
\end{equation}
where $\bZ^t_{\sigma}\sim\normal(0,\Sigma_{t,t}\id_{m})$ and $\bZ^t_{\tau}\sim\normal(0,\Tau_{t,t}\id_n)$. We denote the asymmetric AMP iterates $(\bu^t,\bv^t)_{t\geq 0}$ by $\left\{ \bu^{t},\bv^{t}\middle\vert e_{t},g_{t},\bu^{0}\right\} $.

We are now in position to state our main result.
\begin{thm}\label{thm:AMP_convergence_asym}
Under assumptions \ref{assumption:gaussian-asym}-\ref{assumption:gaussian-conv-asym},
consider the asymmetric AMP iteration $\left\{ \bu^{t},\bv^{t}\middle\vert e_{t},g_{t},\bu^{0}\right\} $
along with its state evolution $\left\{ \Tau_{s,t},\Sigma_{s,t}\middle\vert e_{t},g_{t},\bu^{0}\right\} $.
Define for all $n$, 
\begin{align}
\big(\bZ_{\sigma}^{0},\dots,\bZ_{\sigma}^{t-1}\big) & \sim\normal\big(0,(\Sigma_{s,r})_{0\leq s,r\leq t-1}\otimes \id_{m}\big),\\
\big(\bZ_{\tau}^{1},\dots,\bZ_{\tau}^{t}\big) & \sim\normal\big(0,(\Tau_{s,r})_{1\leq s,r\leq t}\otimes \id_{n}\big),
\end{align}
such that the two collections $(\bZ_{\sigma}^{0},\dots,\bZ_{\sigma}^{t-1})$
and $(\bZ_{\tau}^{1},\dots,\bZ_{\tau}^{t})$ are independent of each other. 
Assume further that $\Sigma_{0,0},\dots,\Sigma_{t-1,t-1},\Tau_{1,1},\dots,\Tau_{t,t}>0$.

Then for any deterministic sequence $\phi_{n}:\left(\R^{n}\times\R^{m}\right)^{t}\times\R^{n}\to\R$ of uniformly pseudo-Lipschitz functions of order $k$, 
\begin{align}
\phi_{n}\big(\bu^{0},\bv^{0},\bu^{1},\bv^{1},\dots ,\bv^{t-1},\bu^{t}\big)
\approxP\E\left[\phi_{n}\left(\bu^{0},\bZ_{\sigma}^{0},\bZ_{\tau}^{1},\bZ_{\sigma}^{1},\dots,\bZ_{\sigma}^{t-1},\bZ_{\tau}^{t}\right)\right]\, .\label{eq:StateEvolutionConvAsymm}
\end{align}
\end{thm}
The proof of this theorem is presented in Section \ref{sec:Asymmetric-AMP}, and is obtained by reduction to the symmetric case, which is treated 
in the next section.

As mentioned above, we use Eq.~(\ref{eq:onsager_asym}) to define the coefficients $\sb_t$, $\sd_t$ because this simplifies the proofs.
In practice, this definition is replaced by an empirical estimate, e.g. as in Eq.~(\ref{eq:FIRSTonsager_asym}). State evolution follows for these versions of AMP
provided such estimates of $\sb_t$, $\sd_t$ are consistent. 

\begin{coro}\label{coro:Empirical}
Consider the modified AMP iteration whereby Eqs.~(\ref{eq:asymmetricAMP_1}), (\ref{eq:asymmetricAMP_2}) are replaced by
\begin{align}
\hbu^{t+1} & =\bA^{\sT}g_{t}(\hbv^{t})-\hsd_{t}e_{t}(\hbu^{t}),\label{eq:asymmetricAMP_1_emp}\\
\hbv^{t} & =\bA e_{t}(\hbu^{t})-\hsb_{t}g_{t-1}(\hbv^{t-1})\, ,\label{eq:asymmetricAMP_2_emp}
\end{align}
with the initialization $\hbu^0=\bu^0$, where  $\hsb_t= \hsb_t(\hbu^0,\hbv^0,\dots, \hbv^{t-1},\hbu^t)$ and 
$\hsd_t = \hsd_t(\hbu^0,\hbv^0,\dots, \hbv^{t-1},\hbu^t,\hbv^t)$ are two estimators of $\sb_t$, $\sd_t$. Assume the same conditions as Theorem \ref{thm:AMP_convergence_asym}. If, for each $t$,
$\hsb_t(\,\cdot\,)$, $\hsd_t(\,\cdot\,)$  are such that
\begin{align}
 \hsb_t(\hbu^0,\hbv^0,\dots, \hbv^{t-1},\hbu^t) \approxP \sb_t\,,\;\;\;\;\;
 \hsd_t(\hbu^0,\hbv^0,\dots, \hbv^{t-1},\hbu^t,\hbv^t) \approxP \sd_t\, ,\label{eq:ConsistentOnsager}
\end{align}
then the iterates $(\hbu^t,\hbv^t)_{t\ge 0}$ satisfy state evolution, namely Eq.~(\ref{eq:StateEvolutionConvAsymm})
holds with $(\bu^t,\bv^t)_{t\ge 0}$ replaced by $(\hbu^t,\hbv^t)_{t\ge 0}$.
\end{coro}
The proof of this statement is deferred to Section \ref{sec:Asymmetric-AMP}.

Two choices of $\hsb_t, \hsd_t$ that satisfy the assumptions are:
\begin{itemize}
\item The empirical values
\begin{align}
\hsb_{t}=\frac{1}{m}\div\, e_{t}(\hbu^{t}) \, ,\;\;\;\hsd_{t}=\frac{1}{m}\div\, g_{t}(\hbv^{t}) \, .
\end{align}
By Theorem \ref{thm:AMP_convergence_asym}, if $\div\, e_{t}(\,\cdot\,)/m$, $\div\, g_{t}(\,\cdot\,)/m$
are uniformly pseudo-Lipschitz, then the assumptions of Corollary \ref{coro:Empirical} hold, and hence
we can apply state evolution. 
\item  As an alternative,
\begin{align}
\hsb_{t}=\frac{n \<\hbu^t,e_t(\hbu^t)\>}{m\|\hbu^t\|_2^2} \, ,\;\;\;\hsd_{t}=\frac{\<\hbv^t,g_t(\hbv^t)\>}{\|\hbv^t\|_2^2} \, .
\end{align}
Consistency follows (for $e_t(\,\cdot\,)$, $g_t(\,\cdot\,)$ uniformly Lipschitz) from Theorem \ref{thm:AMP_convergence_asym}
and Gaussian integration by parts (in particular, Stein's lemma; see Lemma \ref{lem:Stein}).
\end{itemize}

\section{Symmetric AMP}
\label{sec:Symmetric-AMP}
\label{subsec:Symmetric-SE}

For all $n$, we are given a (deterministic) vector
$\bx^{0}\in\R^{n}$ and a sequence of (deterministic) functions $\left\{ f_{t}:\mathbb{R}^{n}\to\mathbb{R}^{n}\right\} _{t\in\mathbb{N}}$.
These will be referred to as the \emph{setting} $\left\{ \bx^{0},f_{t}\right\} $.
Given a sequence of (random) symmetric matrices $\bA=\bA(n)\in\mathbb{R}^{n\times n}$,
we consider the following symmetric AMP iteration
\begin{align}
\bx^{t+1} & =\bA\bm^{t}-\sb_{t}\bm^{t-1}\, ,\label{eq:AMP_1}\\
\bm^{t} & =f_{t}(\bx^{t})\, ,\label{eq:AMP_2}
\end{align}
for $t\in\mathbb{N}$, with  initialization $\bx^0$ (and $\bm^{-1}=0$). Here 
\begin{align}
\sb_{t}=\frac{1}{n}\E\left[\div\, f_{t}(\bZ^{t})\right],\label{eq:Onsager_def}
\end{align}
where $\bZ^{t}\sim\normal(0,\Kappa_{t,t}\id_n)$ and $\Kappa_{t,t}$ will be defined via the state evolution recursion below
(see in particular Eq.~(\ref{eq:state_evolution})). We denote this AMP
recursion as $\left\{ \bx^{t},\bm^{t}\middle\vert f_{t},\bx^{0}\right\} $,
to make explicit the dependence on the setting.

We insist on the fact that $\bA$, $f_{t}$ and $\bx^{0}$ depend on $n$.
However, we will drop this dependence most of the time to ease the reading.

We make the following assumptions.
\begin{enumerate}[font={\bfseries},label={(A\arabic*)}]
\item \label{assumption:gaussian_matrix}$\bA$ is sampled from the Gaussian
orthogonal ensemble $\GOE(n)$, i.e. $\bA=\bG+\bG^{\sT}$ for $\bG\in\reals^{n\times n}$
with i.i.d. entries $G_{ij}\sim\normal(0,1/(2n))$.
\item For each $t\in\N$, $f_{t}:\R^{n}\to\R^{n}$ is uniformly
Lipschitz (as a sequence in $n$). 
\item \label{assumption: x0} $\| \bx^{0}\| _{2}/\sqrt{n}$
converges to a finite constant as $n\to\infty$. 
\item \label{assumption:x^0-conv} The following limit exists and is finite:
\begin{equation}
\Kappa_{1,1}\equiv \lim_{n\rightarrow\infty}\frac{1}{n}\big\< f_{0}(\bx^{0}),f_{0}(\bx^{0})\big\> .
\end{equation}
\item \label{assumption:mixed-conv} For any $t\in\mathbb{N}_{>0}$ and any
$s\ge 0$, the following limit exists and is finite: 
\begin{align}
\lim_{n\to\infty}\frac{1}{n}\E\left[\big\< f_{0}(\bx^{0}),f_{t}(\bZ)\big\> \right]
\end{align}
where $\bZ\in\R^{n}$, $\bZ\sim\normal(0,s\id_{n})$.
\item \label{assumption:gaussian-conv} For any $s,t\in\N_{>0}$ and
any $\bS\in\reals^{2\times 2}$, $\bS\succeq \bzero$, the following
limit exists and is finite: 
\begin{equation}
\lim_{n\rightarrow\infty}\frac{1}{n}\E\left[\big\< f_{s}(\bZ),f_{t}(\bZ')\big\> \right]
\end{equation}
where $(\bZ,\bZ')\in(\R^{n})^{2}$, $(\bZ,\bZ')\sim\normal(0,\bS\otimes \id_{n})$. 
\end{enumerate}
Given assumptions \ref{assumption:x^0-conv}, \ref{assumption:mixed-conv}
and \ref{assumption:gaussian-conv} we can define a doubly infinite array $(\Kappa_{s,r})_{s,r\ge 1}$ via a
\emph{state evolution} recursion as follows.

The initial condition $\Kappa_{1,1}$ is given by assumption \ref{assumption:x^0-conv}.
Once $\bKappa^{(t)} =\left(\Kappa_{s,r}\right)_{s,r\leq t}$ is defined,
let $(\bZ^{1},\dots,\bZ^{t})\sim\normal(0,\bKappa^{(t)}\otimes \id_{n})$
and define, for $0\leq s\leq t$,
\begin{align}
\Kappa_{t+1,s+1} & =\lim_{n\to\infty}\frac{1}{n}\E\left[\big\< f_{s}(\bZ^{s}),f_{t}(\bZ^{t})\big\> \right]\, .\label{eq:state_evolution}
\end{align}
where it is understood that $\bZ^0= \bx^0$ and $\Kappa_{s+1,t+1} = \Kappa_{t+1,s+1}$ is fixed by symmetry.
We will refer to $(\Kappa_{s,t})_{s,t\ge 1}$ as to the state evolution iterates, and we will emphasize their dependence on the setting 
 denoting them by $\left\{ \Kappa_{s,t}\middle\vert f_{t},\bx^{0}\right\}$.
The Onsager term in Eq.~(\ref{eq:AMP_1}) is defined as per Eq.~(\ref{eq:Onsager_def}), with
$\bZ^{t}\sim\normal(0,\Kappa_{t,t}\id_{n})$, and $\Kappa_{t,t}$ given by state evolution.

We have can now state the following state evolution characterization of symmetric AMP, which is analogous to
Theorem \ref{thm:AMP_convergence_asym}.
\begin{thm}
\label{thm:AMP_convergence}Under assumptions \ref{assumption:gaussian_matrix}-\ref{assumption:gaussian-conv},
consider the AMP iteration $\left\{ \bx^{t},\bm^{t}\middle\vert f_{t},\bx^{0}\right\} $.
Define for all $n$, 
\begin{equation}
(\bZ^{1},\dots,\bZ^{t+1})\sim\normal\big(0,\left(\Kappa_{s,r}\right)_{s,r\leq t+1}\otimes \id_{n}\big).
\end{equation}
Assume further that $\Kappa_{1,1},\dots,\Kappa_{t,t}>0$. For any sequence
of uniformly pseudo-Lipschitz functions $\left\{ \phi_{n}:\left(\mathbb{R}^{n}\right)^{t+2}\to\mathbb{R}\right\} $,
\begin{equation}
\phi_{n}\big(\bx^{0},\bx^{1},\dots,\bx^{t+1}\big)\approxP\E\left[\phi_{n}\big(\bx^{0},\bZ^{1},\dots,\bZ^{t+1}\big)\right].
\end{equation}
\end{thm}
The proof of this theorem is presented in Section \ref{sec:ProofSymmetric}. We also note that an analogue of Corollary \ref{coro:Empirical}
applies to this case as well, and $\sb_t$ can be replaced by a consistent estimator $\hsb_t$.

\section{Proof of Theorem \ref{thm:AMP_convergence} (Symmetric AMP)}
\label{sec:ProofSymmetric}

In this section we prove Theorem \ref{thm:AMP_convergence} using a sequence of lemmas, whose proofs are postponed
to Section \ref{subsec:Proof-Lemmas}. We will also try to motivate the main steps. 
Throughout this section  and the next, Assumptions \ref{assumption:gaussian_matrix}-\ref{assumption:gaussian-conv} hold.

\subsection{Notations}
\label{sec:Notations}

We generally denote scalars by lower case letters, e.g. $a$, $b$, $c$, vectors by lower case
boldface, e.g. $\ba$, $\bb$, $\bc$, and matrices by upper case boldface, e.g. $\bA$, $\bB$, $\bC$.
We also use the upper case to emphasize that we are referring to a random variable, and --with a slight abuse of
the convention-- upper case boldface for random vectors.

For two random variables $X$ and $Y$ and a $\sigma$-algebra
$\mathfrak{S}$, we use $X\vert_{\mathfrak{S}}\ed Y$
to mean that for any integrable function $\phi$ and any $\mathfrak{S}$-measurable
bounded random variable $Z$, $\mathbb{E}\left[\phi\left(X\right)Z\right]=\mathbb{E}\left[\phi\left(Y\right)Z\right]$.
In words, $X$ is distributed as $Y$ \textit{conditional on} $\mathfrak{S}$.
If $\mathfrak{S}$ is the trivial $\sigma$-algebra, we simply write $X\ed Y$, i.e. $X$ is distributed
as $Y$.

For two vectors $\bx,\by\in\mathbb{R}^{n}$, we denote their inner product by 
$\< \bx,\by\> =\sum_{i=1}^{n}x_{i}y_{i}$, and  the associated norm by $\Vert\bx\Vert_{2}$.
For two matrices $\bX$, $\bY\in\reals^{m\times n}$, $\<\bX,\bY\> = \Tr(\bX^{\sT}\bY)$ is their scalar product
when viewed as vectors.

We use $\id_{n}$ to denote the $n\times n$ identity matrix. We use
$\sigma_{\min}\left(\bQ\right)$ and $\sigma_{\max}\left(\bQ\right)=\left\Vert \bQ\right\Vert _{\op}$
to denote the minimum and maximum singular values of the matrix $\bQ$.
For two matrices $\bQ$ and $\bP$ of the same number of rows, $\left[\bQ\vert \bP\right]$
denotes the matrix by concatenating $\bQ$ and $\bP$ horizontally. For
any matrix $\bM$, we denote the orthogonal projection onto its range
$\bP_{M}$, and we let $\bP_{M}^{\perp}=\id-\bP_{M}$. When $\bM$ is an empty matrix,
$\bP_{M}=0$ and $\bP_{M}^{\perp}=\id$. When $\bM$ has full column rank,
$\bP_{M}=\bM\left(\bM^{\sT}\bM\right)^{-1}\bM^{\sT}$.

If $f:\R^{n}\to\R^{n}$ is a Lipschitz function, it is almost everywhere
differentiable (w.r.t. the Lebesgue measure), and thus we can define
almost everywhere the quantity 
\begin{equation}
\div \, f(\bx)=\sum_{i=1}^{n}\frac{\partial f_{i}}{\partial x_{i}}(\bx)
\end{equation}
where $f_{i}(\bx)$ is the $i$-th coordinate of $f(\bx)$.

We say that a sequence of events that depends on $n$, hold with
high probability (w.h.p.) if it holds with probability converging
to 1 as $n\to\infty$.

We define the Wasserstein distance (of order $2$) between two probability
measures $\mu$ and $\nu$ as 
\begin{equation}
W_{2}\left(\mu,\nu\right)=\inf_{(X,Y)}\E\left[\left(X-Y\right)^{2}\right]^{1/2},
\end{equation}
where the infimum is taken over all couplings of $\mu$ and $\nu$,
i.e. all random variables $(X,Y)$ such that $X\sim\mu$ and $Y\sim\nu$
marginally.

\subsection{Long AMP}

The main idea of the proof is to analyze a different recursion than
the AMP recursion (\ref{eq:AMP_1}), (\ref{eq:AMP_2}). This new recursion
satisfies the conclusion of Theorem \ref{thm:AMP_convergence} and
will be a good approximation of the AMP recursion in the asymptotic
$n\rightarrow\infty$. It is defined as:
\begin{align}
\bh^{t+1} & =\bP_{Q_{t-1}}^{\perp}\bA\bP_{Q_{t-1}}^{\perp}\bq^{t}+\bH_{t-1}\balpha^{t},\label{eq:new_AMP_1}\\
\bq^{t} & =f_{t}\left(\bh^{t}\right),\label{eq:new_AMP_2}
\end{align}
where at each step $t$, we have defined 
\begin{align}
\bQ_{t-1} & =\left[\bq^{0}\vert \bq^{1}\vert\cdots\vert \bq^{t-1}\right],\label{eq:def_Q}\\
\balpha^{t} & =\left(\bQ_{t-1}^{\sT}\bQ_{t-1}\right)^{-1}\bQ_{t-1}^{\sT}\bq^{t},\\
\bH_{t-1} & =\left[\bh^{1}\vert \bh^{2}\vert\cdots\vert \bh^{t}\right].
\end{align}
The initialization is $\bq^{0}=f_{0}(\bx^{0})$ and $\bh^{1}=\bA\bq^{0}$.
This recursion will be referred as the \emph{Long AMP} recursion,
or LAMP $\left\{ \bh^{t},\bq^{t}\middle\vert f_{t},\bx^{0}\right\} $.

Note that for the LAMP recursion to be well-defined, the matrices
$\bQ_{t-1}^{\sT}\bQ_{t-1}$ must be invertible, that is to say the family
$\bq^{0},\bq^{1},\dots,\bq^{t-1}$ must be linearly independent. This has
no reason to be true, since $\bq^{s}=f_{s}(\bh^{s})$ and $f_{s}$ is
a generic sequence of Lipschitz functions (satisfying assumptions
\ref{assumption:x^0-conv}-\ref{assumption:gaussian-conv}). For instance,
if all $f_{s}$, $s=0,\dots,t-1$, have images included in a same subspace
of dimension lower than $t$, this cannot be true. This difficulty
leads to some technicalities in the proof. However, we will start
by studying the case where $\bQ_{t-1}^{\sT}\bQ_{t-1}$ is invertible, with
$\sigma_{\min}\left(\bQ_{t-1}\right)/\sqrt{n}\geq c_{t}>0$, for $n$
large enough, where $c_{t}$ is a constant independent of $n$. More
formally, we make the following assumption.

\paragraph{Assumption (non-degeneracy):}

We say that the LAMP iterates satisfy the non-degeneracy assumption
if 
\begin{itemize}
\item almost surely, for all $t\in\N$ and all $n\geq t$, $\bQ_{t-1}$ has
full column rank, 
\item for all $t\in\mathbb{N}_{>0}$, there exists some constant $c_{t}>0$
-independent of $n$- such that almost surely, there exists $n_{0}$
(random) such that, for $n\geq n_{0}$, $\sigma_{\min}\left(\bQ_{t-1}\right)/\sqrt{n}\geq c_{t}>0$. 
\end{itemize}

\subsection{The non-degenerate case}

The LAMP recursion is of interest because it behaves well with Gaussian
conditioning, so that the sequence of iterates becomes easier to study.
The following lemma makes this idea explicit. 
\begin{lem}
\label{lem:new_AMP_conditional_distibution}Consider the LAMP $\left\{ \bh^{t},\bq^{t}\middle\vert f_{t},\bx^{0}\right\} $,
and assume it satisfies the non-degeneracy assumption. Fix $t\in\mathbb{N}_{>0}$.
Let $\mathfrak{S}_{t}$ be the $\sigma$-algebra generated by $\bh^{1},\dots,\bh^{t}$
and denote $\bq_{\perp}^{t}=\bP_{Q_{t-1}}^{\perp}\bq^{t}$ and $\bq_{\parallel}^{t}=\bP_{Q_{t-1}}\bq^{t}$.
Then: 
\begin{equation}
\bh^{t+1}\vert_{\mathfrak{S}_{t}}\ed \bP_{Q_{t-1}}^{\perp}\tilde{\bA}\bq_{\perp}^{t}+\bH_{t-1}\balpha^{t}
\end{equation}
where $\tilde{\bA}$ is an independent copy of $\bA$. 
\end{lem}
Here, we decompose $\bh^{t+1}$ as a sum of past iterates $\bh^{1},\dots,\bh^{t}$,
and of a new Gaussian vector $\bP_{Q_{t-1}}^{\perp}\tilde{\bA}\bq_{\perp}^{t}$,
whose conditional law knowing the past $\mathfrak{S}_{t}$ is well
understood. The key property is that we have replaced $\bA$ by a new
matrix $\tilde{\bA}$ decoupled from the past iterates. This enables
us to show that the sets of points $\bq^{0},\bq^{1},\dots,\bq^{t}$ and $\bh^{1},\bh^{2},\dots,\bh^{t+1}$
have asymptotically the same geometry, and that the conclusion of
Theorem \ref{thm:AMP_convergence} holds for LAMP. The following lemma gives a precise statement.
\begin{lem}
\label{lem:new_AMP_convergence} Consider the LAMP $\left\{ \bh^{t},\bq^{t}\middle\vert f_{t},\bx^{0}\right\} $
and suppose it satisfies the non-degeneracy assumption. Then:
\begin{enumerate}[label=(\alph*)]
\item \label{item:same_geometry} For all $0\leq s,r\leq t$, 
\begin{equation}
\frac{1}{n}\left\< \bh^{s+1},\bh^{r+1}\right\> \approxP\frac{1}{n}\left\< \bq^{s},\bq^{r}\right\> .
\end{equation}
\item \label{item:LAMP_SE} For any $t\in\mathbb{N}$, for any sequence
of uniformly order-$k$ pseudo-Lipschitz functions $\left\{ \phi_{n}:\left(\mathbb{R}^{n}\right)^{t+2}\to\mathbb{R}\right\} $,
\begin{equation}
\phi_{n}\left(\bx^{0},\bh^{1},\dots,\bh^{t+1}\right)\approxP\mathbb{E}\left[\phi_{n}\left(\bx^{0},\bZ^{1},\dots,\bZ^{t+1}\right)\right]\label{eq:new_AMP_lemma_h}
\end{equation}
where 
\begin{equation}
(\bZ^{1},\dots,\bZ^{t+1})\sim\normal\left(0,\left(\Kappa_{s,r}\right)_{s,r\leq t+1}\otimes \id_{n}\right).
\end{equation}
Here the state evolution $\{\Kappa_{s,t}\vert f_t, \bx^0\}$ is described in Section \ref{sec:Symmetric-AMP}.
\end{enumerate}
\end{lem}
To conclude that Theorem \ref{thm:AMP_convergence} holds in this
case, we only need to show that LAMP is a good approximation of AMP. 
\begin{lem}
\label{lem:new_AMP_approximates_AMP}Consider the AMP $\left\{ \bx^{t},\bm^{t}\middle\vert f_{t},\bx^{0}\right\} $
and the LAMP $\left\{ \bh^{t},\bq^{t}\middle\vert f_{t},\bx^{0}\right\} $.
Suppose the LAMP satisfies the non-degeneracy assumption. For any
$t\in\mathbb{N}$, 
\begin{equation}
\frac{1}{\sqrt{n}}\left\Vert \bh^{t+1}-\bx^{t+1}\right\Vert _{2}\xrightarrow[n\to\infty]{{\rm P}}0\quad{\rm and}\quad\frac{1}{\sqrt{n}}\left\Vert \bq^{t}-\bm^{t}\right\Vert _{2}\xrightarrow[n\to\infty]{{\rm P}}0.
\end{equation}
\end{lem}
Wrapping things together, we have shown the following weaker form
of Theorem \ref{thm:AMP_convergence}. 
\begin{thm}
\label{thm:AMP_convergence_weak}Assume \ref{assumption:gaussian_matrix}-\ref{assumption:gaussian-conv} and
that the LAMP iterates satisfy the non-degeneracy assumption.
Consider the AMP $\left\{ \bx^{t},\bm^{t}\middle\vert f_{t},\bx^{0}\right\} $.
For any sequence of uniformly order-$k$ pseudo-Lipschitz functions
$\left\{ \phi_{n}:\left(\mathbb{R}^{n}\right)^{t+2}\to\mathbb{R}\right\} $,
\begin{equation}
\phi_{n}\left(\bx^{0},\bx^{1},\dots,\bx^{t+1}\right)\approxP\mathbb{E}\left[\phi_{n}\left(\bx^{0},\bZ^{1},\dots,\bZ^{t+1}\right)\right]
\end{equation}
where 
\begin{equation}
(\bZ^{1},\dots,\bZ^{t+1})\sim\normal\left(0,\left(\Kappa_{s,r}\right)_{s,r\leq t+1}\otimes \id_{n}\right).
\end{equation}
\end{thm}

\subsection{The general case}

To treat the case where the matrix $\bQ_{t-1}$ is ill-conditioned,
we add a small perturbation to the functions $f_{s}$ so that the
perturbed AMP behaves well. We then make sure that the perturbed AMP
approximates well the original one.

A convenient way implement this program is to perturb \emph{randomly} the functions.
We then show that almost surely, the perturbation has the required
properties \ref{assumption:x^0-conv}-\ref{assumption:gaussian-conv}.
Specifically, consider 
\begin{equation}
f_{t}^{\epsilon y}\left(\,\cdot\,\right)=f_{t}\left(\,\cdot\,\right)+\epsilon \by^{t}
\end{equation}
where $\epsilon\geq0$ and $\by^{0},\by^{1},\by^{2},\dots$ are generated
as i.i.d. $\normal\left(0,\id_{n}\right)$, independent of the matrix
$A$. The perturbation vectors $\by^{0},\by^{1},\by^{2},\dots$ are called
collectively as $\by$ for brevity. 
\begin{lem}
\label{lem:SE_perturbated} Almost surely (w.r.t. $\by$), the setting
$\left\{ \bx^{0},f_{t}^{\epsilon y}\right\} $ satisfies assumptions
\ref{assumption:x^0-conv},\ref{assumption:mixed-conv}, \ref{assumption:gaussian-conv}.
As a consequence, we can define an associated state evolution $\left\{ \Kappa_{s,t}^{\epsilon}\middle\vert f_{t}^{\epsilon y},\bx^{0}\right\} $:
\begin{equation}
\Kappa_{1,1}^{\epsilon}=\lim_{n\to\infty}\frac{1}{n}\left\Vert f_{0}^{\epsilon y}\left(\bx^{0}\right)\right\Vert _{2}^{2},
\end{equation}
and once $\Kappa^{\epsilon}=\left(\Kappa_{s,r}^{\epsilon}\right)_{s,r\leq t}$
is defined, take $\left(\bZ^{\epsilon,1},\dots,\bZ^{\epsilon,t}\right)\sim\normal\left(0,\Kappa^{\epsilon}\otimes \id_{n}\right)$ independently of $\by$
and define 
\begin{align}
\Kappa_{1,t+1}^{\epsilon} & =\lim_{n\to\infty}\frac{1}{n}\mathbb{E}\left[\left\< f_{0}^{\epsilon y}\left(\bx^{0}\right),f_{t}^{\epsilon y}\left(\bZ^{\epsilon,t}\right)\right\> \right],\\
\Kappa_{s+1,t+1}^{\epsilon} & =\lim_{n\to\infty}\frac{1}{n}\mathbb{E}\left[\left\< f_{s}^{\epsilon y}\left(\bZ^{\epsilon,s}\right),f_{t}^{\epsilon y}\left(\bZ^{\epsilon,t}\right)\right\> \right.],
\end{align}
where the expectations are taken w.r.t. $\bZ^{\epsilon,1},\dots,\bZ^{\epsilon,t}$ but not $\by$.
Moreover, the resulting state evolution is almost surely equal to
a constant, thus justifying that we drop the dependence on $\by$ in
$\Kappa_{s,t}^{\epsilon}$. 
\end{lem}
\begin{lem}
\label{lem:good-cond-perturbed}Denote as $\bQ_{t-1}^{\epsilon y}$
the matrix associated with the LAMP iterates $\left\{ \bh^{\epsilon y,t},\bq^{\epsilon y,t}\middle\vert f_{t}^{\epsilon y},\bx^{0}\right\} $,
according to equation (\ref{eq:def_Q}). Assume $\epsilon>0$. Then
as soon as $n\geq t$, almost surely the matrix $\bQ_{t-1}^{\epsilon y}$
is of full column rank. Furthermore, there exists a constant $c_{t,\epsilon}>0$
-independent of $n$- such that almost surely, there exists $n_{0}$
(random) such that for $n\geq n_{0}$, $\sigma_{\min}(\bQ_{t-1}^{\epsilon y})/\sqrt{n}\geq c_{t,\epsilon}$. 
\end{lem}
The last two lemmas imply that almost surely, we can apply Theorem
\ref{thm:AMP_convergence_weak} to $\left\{ f_{t}^{\epsilon y}\right\} _{t\geq0}$.
The next three lemmas quantify how this result approximates our original
one. 
\begin{lem}
\label{lem:pseudo-lip_and_gaussians}Let $\left\{ \phi_{n}:\left(\R^{n}\right)^{t}\rightarrow\R\right\} _{n>0}$
be a sequence of uniformly pseudo-Lipschitz functions of order $k$.
Let $\Kappa$, $\tilde{\Kappa}$ be two $t\times t$ covariance matrices,
and $\bZ\sim\normal\left(0,\Kappa\otimes \id_{n}\right)$, $\tilde{\bZ}\sim\normal\left(0,\tilde{\Kappa}\otimes \id_{n}\right)$.
Then 
\begin{equation}
\lim_{\tilde{\Kappa}\to\Kappa}\sup_{n\geq1}\left\vert \E[\phi_{n}(\bZ)]-\E[\phi_{n}(\tilde{\bZ})]\right\vert =0.
\end{equation}
\end{lem}
\begin{lem}
\label{lem:SE_eps_to_zero}For any $s,t\geq1$, $\Kappa_{s,t}^{\epsilon}\xrightarrow[\epsilon\to0]{}\Kappa_{s,t}$. 
\end{lem}
\begin{lem}
\label{lem:AMP_eps_to_0}Consider the AMP iterates in two different
settings $\left\{ \bx^{t},\bm^{t}\middle\vert f_{t},\bx^{0}\right\} $ and
$\left\{ \bx^{\epsilon y,t},\bm^{\epsilon y,t}\middle\vert f_{t}^{\epsilon y},\bx^{0}\right\} $.
Assume further that for some $t\in\N$, $\Kappa_{1,1},\dots,\Kappa_{t,t}>0$.
Then there exist functions $h_{t}(\epsilon)$, $h'_{t}(\epsilon)$,
independent of $n$, such that 
\begin{equation}
\lim_{\epsilon\to0}h_{t}(\epsilon)=\lim_{\epsilon\to0}h_{t}'(\epsilon)=0,
\end{equation}
and for all $\epsilon\leq1$, with high probability, 
\begin{align}
\frac{1}{\sqrt{n}}\left\Vert \bm^{\epsilon y,t}-\bm^{t}\right\Vert _{2} & \leq h_{t}'(\epsilon),\\
\frac{1}{\sqrt{n}}\left\Vert \bx^{\epsilon y,t+1}-\bx^{t+1}\right\Vert _{2} & \leq h_{t}(\epsilon).
\end{align}
\end{lem}
\begin{proof}[Proof of Theorem \ref{thm:AMP_convergence}]
The proof combines three elements that follow from the previous lemmas: 
\begin{itemize}
\item Thanks to Lemmas \ref{lem:SE_perturbated} and \ref{lem:good-cond-perturbed},
\emph{almost surely w.r.t. the perturbation} $\by^{0},\by^{1},\dots$, the
assumptions of Theorem \ref{thm:AMP_convergence_weak} are satisfied
for the perturbed setting $\left\{ \bx^{0},f_{t}^{\epsilon y}\right\} $.
We get that a.s., for any sequence of uniformly pseudo-Lipschitz functions
$\left\{ \phi_{n}:\left(\mathbb{R}^{n}\right)^{t+2}\to\mathbb{R}\right\} $,
\begin{equation}
\phi_{n}\left(\bx^{0},\bx^{\epsilon y,1},\dots,\bx^{\epsilon y,t+1}\right)\approxP\mathbb{E}\left[\phi_{n}\left(\bx^{0},\bZ^{\epsilon,1},\dots,\bZ^{\epsilon,t+1}\right)\right]\label{eq:thm_5_applied}
\end{equation}
where $\bZ^{\epsilon,1},\dots,\bZ^{\epsilon,t+1}\sim\normal\left(0,\left(\Kappa_{r,s}^{\epsilon}\right)_{r,s\leq t+1}\otimes \id_{n}\right)$.
To obtain the desired result, we shall take the limit $\epsilon\to0$,
the technicalities of which are presented in the following two elements. 
\item Let $\bZ^{1},\dots,\bZ^{t+1}\sim\normal\left(0,\left(\Kappa_{r,s}\right)_{r,s\leq t+1}\otimes \id_{n}\right)$.
Since, by Lemma \ref{lem:SE_eps_to_zero}, the perturbed state evolution
converges to the original one when $\epsilon\to0$, so we can apply
Lemma \ref{lem:pseudo-lip_and_gaussians} to get 
\begin{equation}
\sup_{n\geq1}\left\vert \mathbb{E}\left[\phi_{n}\left(\bx^{0},\bZ^{\epsilon,1},\dots,\bZ^{\epsilon,t+1}\right)\right]-\mathbb{E}\left[\phi_{n}\left(\bx^{0},\bZ^{1},\dots,\bZ^{t+1}\right)\right]\right\vert \xrightarrow[\epsilon\to0]{}0.\label{eq:gaussian_approx}
\end{equation}
\item Using that $\phi_{n}$ is uniformly pseudo-Lipschitz of order $k$
and the triangle inequality, 
\begin{align}
 & \left\vert \phi_{n}\left(\bx^{0},\bx^{1},\dots,\bx^{t+1}\right)-\phi_{n}\left(\bx^{0},\bx^{\epsilon y,1},\dots,\bx^{\epsilon y,t+1}\right)\right\vert \\
 & \leq LC_{1}(k,t)\left(1+\frac{\left\Vert \bx^{0}\right\Vert _{2}^{k-1}}{n^{(k-1)/2}}+\sum_{i=1}^{t+1}\frac{\left\Vert \bx^{\epsilon y,i}\right\Vert _{2}^{k-1}}{n^{(k-1)/2}}+\sum_{i=1}^{t+1}\frac{\left\Vert \bx^{i}\right\Vert _{2}^{k-1}}{n^{(k-1)/2}}\right)\sum_{i=1}^{t+1}\frac{\left\Vert \bx^{\epsilon y,i}-\bx^{i}\right\Vert _{2}}{\sqrt{n}}
\end{align}
where here $C_{j}(k,t)$ is a constant depending only on $k$ and
$t$. Lemma \ref{lem:AMP_eps_to_0} ensures that w.h.p. $\left\Vert \bx^{\epsilon y,i}-\bx^{i}\right\Vert _{2}/\sqrt{n}\leq h_{i}(\epsilon)$.
We also know by assumption \ref{assumption: x0} that $\left\Vert \bx^{0}\right\Vert _{2}/\sqrt{n}$
converges to a finite limit. Furthermore, one can use Theorem \ref{thm:AMP_convergence_weak}
to bound w.h.p. 
\begin{equation}
\frac{\left\Vert \bx^{\epsilon y,i}\right\Vert _{2}^{k-1}}{n^{(k-1)/2}}=\frac{\E\left[\left\Vert \bZ^{\epsilon,i}\right\Vert _{2}^{k-1}\right]}{n^{(k-1)/2}}+\left(\frac{\left\Vert \bx^{\epsilon y,i}\right\Vert _{2}^{k-1}}{n^{(k-1)/2}}-\frac{\E\left[\left\Vert \bZ^{\epsilon,i}\right\Vert _{2}^{k-1}\right]}{n^{(k-1)/2}}\right)\leq C_{2}(k)\left\Vert \left(\Kappa_{s,r}^{\epsilon}\right)_{s,r\leq t+1}\right\Vert _{{\rm op}}^{(k-1)/2}+1.
\end{equation}
Finally, using the triangle inequality, w.h.p., 
\begin{align}
\frac{\left\Vert \bx^{i}\right\Vert _{2}^{k-1}}{n^{(k-1)/2}} & \leq C_{3}(k)\left(\frac{\left\Vert \bx^{\epsilon y,i}\right\Vert _{2}^{k-1}}{n^{(k-1)/2}}+\frac{\left\Vert \bx^{\epsilon y,i}-\bx^{i}\right\Vert _{2}^{k-1}}{n^{(k-1)/2}}\right)\\
 & \leq C_{4}(k)\left(\left\Vert \left(\Kappa_{s,r}^{\epsilon}\right)_{s,r\leq t+1}\right\Vert _{{\rm op}}^{(k-1)/2}+1+h_{i}(\epsilon)^{k-1}\right).
\end{align}
Putting things together, we get w.h.p., 
\begin{align}
 & \left\vert \phi_{n}\left(\bx^{0},\bx^{1},\dots,\bx^{t+1}\right)-\phi_{n}\left(\bx^{0},\bx^{\epsilon y,1},\dots,\bx^{\epsilon y,t+1}\right)\right\vert \\
 & \leq LC_{5}(k,t)\left(1+\left\Vert \left(\Kappa_{s,r}^{\epsilon}\right)_{s,r\leq t+1}\right\Vert _{{\rm op}}^{(k-1)/2}+\sum_{i=1}^{t+1}h_{i}(\epsilon)^{k-1}\right)\sum_{i=1}^{t+1}h_{i}(\epsilon).
\end{align}
As this upper bound goes converges to 0 as $\epsilon\to0$, we have
for any $\eta>0$, 
\begin{equation}
\lim_{\epsilon\to0}\limsup_{n\to\infty}\Pr\left(\left\vert \phi_{n}\left(\bx^{0},\bx^{1},\dots,\bx^{t+1}\right)-\phi_{n}\left(\bx^{0},\bx^{\epsilon y,1},\dots,\bx^{\epsilon y,t+1}\right)\right\vert \geq\eta\right)=0.\label{eq:phi_approx}
\end{equation}
\end{itemize}
Let us now combine the three elements together. Let $\eta>0$. We
have: 
\begin{align}
 & \Pr\left(\left\vert \phi_{n}\left(\bx^{0},\bx^{1},\dots,\bx^{t+1}\right)-\mathbb{E}\left[\phi_{n}\left(\bx^{0},\bZ^{1},\dots,\bZ^{t+1}\right)\right]\right\vert \geq\eta\right)\\
 & \leq\Pr\left(\left\vert \phi_{n}\left(\bx^{0},\bx^{1},\dots,\bx^{t+1}\right)-\phi_{n}\left(\bx^{0},\bx^{\epsilon y,1},\dots,\bx^{\epsilon y,t+1}\right)\right\vert \geq\frac{\eta}{3}\right)\\
 & \quad+\Pr\left(\left\vert \phi_{n}\left(\bx^{0},\bx^{\epsilon y,1},\dots,\bx^{\epsilon y,t+1}\right)-\mathbb{E}\left[\phi_{n}\left(\bx^{0},\bZ^{\epsilon,1},\dots,\bZ^{\epsilon,t+1}\right)\right]\right\vert \geq\frac{\eta}{3}\right)\\
 & \quad+\mathbbm{1}_{\left\{ \left\vert \E\left[\phi_{n}\left(\bx^{0},\bZ^{\epsilon,1},\dots,\bZ^{\epsilon,t+1}\right)\right]-\E\left[\phi_{n}\left(\bx^{0},\bZ^{1},\dots,\bZ^{t+1}\right)\right]\right\vert \geq\eta/3\right\} }
\end{align}
Taking $\limsup$ as $n\to\infty$, the second term vanishes because
of (\ref{eq:thm_5_applied}): 
\begin{align}
 & \limsup_{n\to\infty}\Pr\left(\left\vert \phi_{n}\left(\bx^{0},\bx^{1},\dots,\bx^{t+1}\right)-\mathbb{E}\left[\phi_{n}\left(\bx^{0},\bZ^{1},\dots,\bZ^{t+1}\right)\right]\right\vert \geq\eta\right)\\
 & \leq\limsup_{n\to\infty}\Pr\left(\left\vert \phi_{n}\left(\bx^{0},\bx^{1},\dots,\bx^{t+1}\right)-\phi_{n}\left(\bx^{0},\bx^{\epsilon y,1},\dots,\bx^{\epsilon y,t+1}\right)\right\vert \geq\frac{\eta}{3}\right)\\
 & \quad+\mathbbm{1}_{\left\{ \sup_{n\geq1}\left\vert \E\left[\phi_{n}\left(\bx^{0},\bZ^{\epsilon,1},\dots,\bZ^{\epsilon,t+1}\right)\right]-\E\left[\phi_{n}\left(\bx^{0},\bZ^{1},\dots,\bZ^{t+1}\right)\right]\right\vert \geq\eta/3\right\} }.
\end{align}
Because of (\ref{eq:phi_approx}) and (\ref{eq:gaussian_approx}),
this upper bound converges to 0 as $\epsilon\to0$. We can then conclude
that 
\begin{equation}
\left\vert \phi_{n}\left(\bx^{0},\bx^{1},\dots,\bx^{t+1}\right)-\mathbb{E}\left[\phi_{n}\left(\bx^{0},\bZ^{1},\dots,\bZ^{t+1}\right)\right]\right\vert \xrightarrow[n\to\infty]{P}0.
\end{equation}
\end{proof}

\subsection{Proof of the Lemmas\label{subsec:Proof-Lemmas}}

\subsubsection{Proof of Lemma \ref{lem:new_AMP_conditional_distibution}}

The claim for $t=0$ is immediate from that $\mathfrak{S}_{0}$ is
the trivial $\sigma$-algebra and $\bP_{Q_{t-1}}^{\perp}=\id_{n}$. For
$t\geq1$, let us rewrite (\ref{eq:new_AMP_1}) as 
\begin{align}
\bh^{t+1} & =\bA\bq_{\perp}^{t}-\bP_{Q_{t-1}}\bA\bq_{\perp}^{t}+\bH_{t-1}\balpha^{t}\\
 & =\bA\bq_{\perp}^{t}-\tilde{\bQ}_{t-1}\left(\tilde{\bQ}_{t-1}^{\sT}\tilde{\bQ}_{t-1}\right)^{-1}\bY_{t-1}^{\sT}\bq_{\perp}^{t}+\bH_{t-1}\balpha^{t}
\end{align}
where $\bq_{\perp}^{t}=\bP_{Q_{t-1}}^{\perp}\bq^{t}$, $\tilde{\bQ}_{t-1}=\left[\bq_{\perp}^{0}\vert \bq_{\perp}^{1}\vert\dots\vert \bq_{\perp}^{t-1}\right]$
and $\bY_{t-1}=\left[\by^{0}\vert \by^{1}\vert\dots\vert \by^{t-1}\right]$
with $\by^{s}=\bA^{\sT}\bq_{\perp}^{s}=\bA\bq_{\perp}^{s}$. Here we use the fact
that $\bP_{Q_{t-1}}=\bP_{\tilde{Q}_{t-1}}$. Notice that 
\begin{align}
\by^{0} & =\bh^{1},\\
\by^{s} & =\bh^{s+1}+\tilde{\bQ}_{s-1}\left(\tilde{\bQ}_{s-1}^{\sT}\tilde{\bQ}_{s-1}\right)^{-1}\bY_{s-1}^{\sT}\bq_{\perp}^{s}-\bH_{s-1}\balpha^{s}
\end{align}
for any $s\geq1$. Also, $\bH_{s-1}$, $\bQ_{s-1}$, and $\tilde{\bQ}_{s-1}$
are $\mathfrak{S}_{t}$-measurable for $1\leq s\leq t$. Then a simple
induction yields that $\bY_{t-1}$ is $\mathfrak{S}_{t}$-measurable.
Hence to find $\bA\vert_{\mathfrak{S}_{t}}$, conditioning on $\mathfrak{S}_{t}$
is equivalent to conditioning on the linear constraint $\bA\tilde{\bQ}_{t-1}=\bY_{t-1}$.
As shown in \cite[Lemma 3]{Javanmard2013} and \cite[Lemma 10]{Bayati2011},
$\bA\vert_{\mathfrak{S}_{t}}\ed\mathbb{E}\left[\bA\middle|\mathfrak{S}_{t}\right]+{\cal P}_{t}\left(\tilde{\bA}\right)$,
where $\tilde{\bA}\ed A$ independent of $\mathfrak{S}_{t}$
and ${\cal P}_{t}$ is the orthogonal projector onto the subspace
$\left\{ \hat{\bA}\in\mathbb{R}^{n\times n}\middle|\hat{\bA}\tilde{\bQ}_{t-1}=0,\hat{\bA}=\hat{\bA}^{\sT}\right\} $:
\begin{align}
\mathbb{E}\left[\bA\middle|\mathfrak{S}_{t}\right] & =\bA-\bP_{\tilde{Q}_{t-1}}^{\perp}\bA\bP_{\tilde{Q}_{t-1}}^{\perp}=\bA-\bP_{Q_{t-1}}^{\perp}\bA\bP_{Q_{t-1}}^{\perp},\\
{\cal P}_{t}\left(\tilde{\bA}\right) & =\bP_{\tilde{Q}_{t-1}}^{\perp}\tilde{\bA}\bP_{\tilde{Q}_{t-1}}^{\perp}=\bP_{Q_{t-1}}^{\perp}\tilde{\bA}\bP_{Q_{t-1}}^{\perp},
\end{align}
where we use $\bP_{\tilde{Q}_{t-1}}^{\perp}=\bP_{Q_{t-1}}^{\perp}$. Then
from (\ref{eq:new_AMP_1}), 
\begin{equation}
\bh^{t+1}\vert_{\mathfrak{S}_{t}}\ed \bP_{Q_{t-1}}^{\perp}\tilde{\bA}\bP_{Q_{t-1}}^{\perp}\bq^{t}+\bH_{t-1}\balpha^{t}
\end{equation}
since $\bP_{Q_{t-1}}^{\perp}\mathbb{E}\left[\bA\middle|\mathfrak{S}_{t}\right]\bP_{Q_{t-1}}^{\perp}=\bP_{Q_{t-1}}^{\perp}\left(A-\bP_{Q_{t-1}}^{\perp}\bA\bP_{Q_{t-1}}^{\perp}\right)\bP_{Q_{t-1}}^{\perp}=0$.

\subsubsection{Proof of Lemma \ref{lem:new_AMP_convergence}}

We prove the results by induction over $t\in\mathbb{N}$. Let the
statement for $t$ be ${\cal H}_{t}$.

\textbf{Proof of ${\cal H}_{0}$.} Recall that $\bh^{1}=\bA\bq^{0}$. Then
\ref{item:same_geometry} follows immediately from Lemma \ref{lem:GOE_properties},
and \ref{item:LAMP_SE} is from Lemmas \ref{lem:GOE_properties},
\ref{lem:pseudoLipschitz_add_constVec}, \ref{lem:psedoLipschitz_concentration}.

\textbf{Proof of ${\cal H}_{t}$.} We assume ${\cal H}_{0},\dots,{\cal H}_{t-1}$
hold and prove ${\cal H}_{t}$. First note that $\balpha^{t}\xrightarrow[n\to\infty]{{\rm P}}\balpha^{t,*}$
a constant vector in $\mathbb{R}^{t}$, using ${\cal H}_{t-1}\left(b\right)$,
Lemma \ref{lem:pseudoLipschitz_innerProd} and the non-degeneracy
assumption.
\begin{enumerate}[label=(\alph*)]
\item We only need to prove the claim for $r=t$.

Consider the case $s<t$. Since $\bh^{s+1}$ and $\left\< \bq^{s},\bq^{t}\right\> $
are $\mathfrak{S}_{t}$-measurable, by Lemma \ref{lem:new_AMP_conditional_distibution},
\begin{equation}
\left.\left(\left\< \bh^{s+1},\bh^{t+1}\right\> -\left\< \bq^{s},\bq^{t}\right\> \right)\right\vert_{\mathfrak{S}_{t}}\ed\left\< \bP_{Q_{t-1}}^{\perp}\bh^{s+1},\tilde{\bA}\bq_{\perp}^{t}\right\> +\left\< \bH_{t-1}^{\sT}\bh^{s+1},\balpha^{t}\right\> -\left\< \bq^{s},\bq^{t}\right\> .
\end{equation}
Note that by ${\cal H}_{t-1}\left(a\right)$, $\frac{1}{n}\left\Vert \bH_{t-1}^{\sT}\bh^{s+1}-\bQ_{t-1}^{\sT}\bq^{s}\right\Vert _{2}\xrightarrow[n\to\infty]{{\rm P}}0$
. Hence, 
\begin{align}
\frac{1}{n}\left|\left\< \bH_{t-1}^{\sT}\bh^{s+1},\balpha^{t}\right\> -\left\< \bq^{s},\bq^{t}\right\> \right| & =\frac{1}{n}\left|\left\< \bH_{t-1}^{\sT}\bh^{s+1},\balpha^{t}\right\> -\left\< \bP_{Q_{t-1}}\bq^{s},\bq^{t}\right\> \right|\\
 & =\frac{1}{n}\left|\left\< \bH_{t-1}^{\sT}\bh^{s+1}-\bQ_{t-1}^{\sT}\bq^{s},\balpha^{t}\right\> \right|\\
 & \leq\frac{1}{n}\left\Vert \bH_{t-1}^{\sT}\bh^{s+1}-\bQ_{t-1}^{\sT}\bq^{s}\right\Vert _{2}\left\Vert \balpha^{t}\right\Vert _{2}\xrightarrow[n\to\infty]{{\rm P}}0
\end{align}
where we use $\balpha^{t}\xrightarrow[n\to\infty]{{\rm P}}\balpha^{t,*}$ (which holds by the induction hypothesis).
Furthermore, since $\tilde{\bA}$ is independent of $\bq_{\perp}^{t}$
and $\bP_{Q_{t-1}}^{\perp}\bh^{s+1}$, by Lemma \ref{lem:GOE_properties},
\begin{equation}
\frac{1}{n}\left|\left\< \bP_{Q_{t-1}}^{\perp}\bh^{s+1},\tilde{\bA}\bq_{\perp}^{t}\right\> \right|\xrightarrow[n\to\infty]{{\rm P}}0
\end{equation}
since $\frac{1}{\sqrt{n}}\left\Vert \bh^{s+1}\right\Vert _{2}$ and
$\frac{1}{\sqrt{n}}\left\Vert \bq^{t}\right\Vert _{2}$ concentrate
at finite constants by ${\cal H}_{t-1}\left(b\right)$ and Lemma \ref{lem:pseudoLipschitz_innerProd},
and $\left\Vert \bP_{Q_{t-1}}^{\perp}\bh^{s+1}\right\Vert _{2}\leq\left\Vert \bh^{s+1}\right\Vert _{2}$,
$\left\Vert \bq_{\perp}^{t}\right\Vert _{2}\leq\left\Vert \bq^{t}\right\Vert _{2}$.
It follows that $\frac{1}{n}\left\< \bh^{s+1},\bh^{t+1}\right\> \approxP\frac{1}{n}\left\< \bq^{s},\bq^{t}\right\> $.

Consider the case $s=t$. Since $\bq^{t}$ is $\mathfrak{S}_{t}$-measurable,
by Lemma \ref{lem:new_AMP_conditional_distibution}, 
\begin{align}
\left(\left\Vert \bh^{t+1}\right\Vert _{2}^{2}-\left\Vert \bq^{t}\right\Vert _{2}^{2}\right)\bigg\vert_{\mathfrak{S}_{t}} & \ed\left\Vert \bP_{Q_{t-1}}^{\perp}\tilde{\bA}\bq_{\perp}^{t}\right\Vert _{2}^{2}+2\left\< \bP_{Q_{t-1}}^{\perp}\tilde{\bA}\bq_{\perp}^{t},\bH_{t-1}\balpha^{t}\right\> +\left\Vert \bH_{t-1}\balpha^{t}\right\Vert _{2}^{2}-\left\Vert \bq^{t}\right\Vert _{2}^{2}\\
 & =\left\Vert \bP_{Q_{t-1}}^{\perp}\tilde{\bA}\bq_{\perp}^{t}\right\Vert _{2}^{2}+2\left\< \tilde{\bA}\bq_{\perp}^{t},\bP_{Q_{t-1}}^{\perp}\bH_{t-1}\balpha^{t}\right\> +\left\< \balpha^{t},\bH_{t-1}^{\sT}\bH_{t-1}\balpha^{t}\right\>-\left\Vert \bq^{t}\right\Vert _{2}^{2}.
\end{align}
Again, $\frac{1}{n}\left\< \tilde{\bA}\bq_{\perp}^{t},\bP_{Q_{t-1}}^{\perp}\bH_{t-1}\balpha^{t}\right\> \xrightarrow[n\to\infty]{{\rm P}}0$.
By independence of $\tilde{\bA}$ and Lemma \ref{lem:GOE_properties},
we get 
\begin{equation}
\frac{1}{n}\left\Vert \bP_{Q_{t-1}}^{\perp}\tilde{\bA}\bq_{\perp}^{t}\right\Vert _{2}^{2}=\frac{1}{n}\left\Vert \tilde{\bA}\bq_{\perp}^{t}\right\Vert _{2}^{2}-\frac{1}{n}\left\Vert \bP_{Q_{t-1}}\tilde{\bA}\bq_{\perp}^{t}\right\Vert _{2}^{2}\approxP\frac{1}{n}\left\Vert \bq_{\perp}^{t}\right\Vert _{2}^{2}.
\end{equation}
Using ${\cal H}_{t-1}\left(a\right)$ and that $\balpha^{t}\xrightarrow[n\to\infty]{{\rm P}}\balpha^{t,*}$,
\begin{equation}
\frac{1}{n}\left\< \balpha^{t},\left(\bH_{t-1}^{\sT}\bH_{t-1}-\bQ_{t-1}^{\sT}\bQ_{t-1}\right)\balpha^{t}\right\> \xrightarrow[n\to\infty]{{\rm P}}0.
\end{equation}
Notice that $\left\< \balpha^{t},\bQ_{t-1}^{\sT}\bQ_{t-1}\balpha^{t}\right\> =\left\Vert \bq_{\parallel}^{t}\right\Vert _{2}^{2}$.
The claim is proven. 
\item First note that $\frac{1}{n}\left\Vert \bh^{t+1}\right\Vert _{2}^{2}\approxP\frac{1}{n}\left\Vert \bq^{t}\right\Vert _{2}^{2}\xrightarrow[n\to\infty]{{\rm P}}\Kappa_{t+1,t+1}$
by ${\cal H}_{t-1}$. By Lemma \ref{lem:new_AMP_conditional_distibution},
\begin{equation}
\left.\phi_{n}\left(\bx^{0},\bh^{1},\dots,\bh^{t},\bh^{t+1}\right)\right\vert_{\mathfrak{S}_{t}}\ed\phi_{n}\left(\bx^{0},\bh^{1},\dots,\bh^{t},\tilde{\bA}\bq_{\perp}^{t}-\bP_{Q_{t-1}}\tilde{\bA}\bq_{\perp}^{t}+\bH_{t-1}\balpha^{t}\right),
\end{equation}
and we denote the right-hand side by $\phi'_{n}\left(\tilde{\bA}\bq_{\perp}^{t}-\bP_{Q_{t-1}}\tilde{\bA}\bq_{\perp}^{t}+\bH_{t-1}\balpha^{t}\right)$
for brevity. Note that $\phi'_{n}$ is uniformly pseudo-Lipschitz by Lemma \ref{lem:pseudoLipschitz_add_constVec} and the induction hypothesis, whence
\begin{equation}
\begin{split} & \left|\phi'_{n}\left(\tilde{\bA}\bq_{\perp}^{t}-\bP_{Q_{t-1}}\tilde{\bA}\bq_{\perp}^{t}+\bH_{t-1}\balpha^{t}\right)-\phi'_{n}\left(\tilde{\bA}\bq_{\perp}^{t}+\bH_{t-1}\balpha^{t}\right)\right|\nonumber\\
 & \leq L_{n}C(k,t)\left\{ 1+\left(\frac{\left\Vert \bx^{0}\right\Vert _{2}}{\sqrt{n}}\right)^{k-1}+\sum_{s=1}^{t}\left(\frac{\left\Vert \bh^{s}\right\Vert _{2}}{\sqrt{n}}\right)^{k-1}\right.\\
 & \hskip6em\relax\left.+\left(\frac{\left\Vert \bh^{t+1}\right\Vert _{2}}{\sqrt{n}}\right)^{k-1}+\left(\frac{\left\Vert \tilde{\bA}\bq_{\perp}^{t}\right\Vert _{2}}{\sqrt{n}}\right)^{k-1}+\left(\frac{\left\Vert \bH_{t-1}\balpha^{t}\right\Vert _{2}}{\sqrt{n}}\right)^{k-1}\right\} \frac{\left\Vert \bP_{Q_{t-1}}\tilde{\bA}\bq_{\perp}^{t}\right\Vert _{2}}{\sqrt{n}}
\end{split}
\end{equation}
where $C(k,t)$ is a constant depending only on $k$ and $t$. We
have: 
\begin{equation}
\frac{1}{\sqrt{n}}\left\Vert \bH_{t-1}\balpha^{t}\right\Vert _{2}\leq\left\Vert \frac{1}{\sqrt{n}}\bH_{t-1}\right\Vert _{2}\left\Vert \balpha^{t}\right\Vert _{2}\leq\left\Vert \frac{1}{\sqrt{n}}\bH_{t-1}\right\Vert _{{\rm F}}\left\Vert \balpha^{t}\right\Vert _{2}=\sqrt{\frac{1}{n}\sum_{s=1}^{t}\left\Vert \bh^{s}\right\Vert _{2}^{2}}\cdot\left\Vert \balpha^{t}\right\Vert _{2}
\end{equation}
which converges to a finite constant by ${\cal H}_{t-1}\left(b\right)$
and that $\balpha^{t}\xrightarrow[n\to\infty]{{\rm P}}\balpha^{t,*}$.
We also have $\frac{1}{\sqrt{n}}\left\Vert \tilde{\bA}\bq_{\perp}^{t}\right\Vert _{2}\leq\frac{1}{\sqrt{n}}\left\Vert \tilde{\bA}\right\Vert _{2}\left\Vert \bq^{t}\right\Vert _{2}$,
which converges to a finite constant due to ${\cal H}_{t-1}\left(b\right)$
and Theorem \ref{thm:A_norm}. Furthermore, by independence of $\tilde{\bA}$,
recalling ${\rm rank}\left(\bP_{Q_{t-1}}\right)\leq t$, we have $\frac{1}{\sqrt{n}}\left\Vert \bP_{Q_{t-1}}\tilde{\bA}\bq_{\perp}^{t}\right\Vert _{2}\xrightarrow[n\to\infty]{{\rm P}}0$
by Lemma \ref{lem:GOE_properties}. Therefore, 
\begin{align}
\phi'_{n}\left(\tilde{\bA}\bq_{\perp}^{t}-\bP_{Q_{t-1}}\tilde{\bA}\bq_{\perp}^{t}+\bH_{t-1}\balpha^{t}\right) & \approxP\phi'_{n}\left(\tilde{\bA}\bq_{\perp}^{t}+\bH_{t-1}\balpha^{t}\right)\\
 & \approxP\phi'_{n}\left(\tilde{\bA}\bq_{\perp}^{t}+\bH_{t-1}\balpha^{t,*}\right).
\end{align}
Notice that $\frac{1}{n}\left\Vert \bq_{\perp}^{t}\right\Vert _{2}^{2}=\frac{1}{n}\left\Vert \bq^{t}\right\Vert _{2}^{2}-\frac{1}{n}\left\< \balpha^{t},\bQ_{t-1}^{\sT}\bQ_{t-1}\balpha^{t}\right\> $,
which converges to a constant $a^{2}$ due to ${\cal H}_{t-1}\left(b\right)$
and that $\balpha^{t}\xrightarrow[n\to\infty]{{\rm P}}\balpha^{t,*}$.
Then by Lemma \ref{lem:GOE_properties}, there exists $\bZ\sim\normal\left(0,\id_{n}\right)$
independent of $\mathfrak{S}_{t}$ such that 
\begin{align}
\phi'_{n}\left(\tilde{\bA}\bq_{\perp}^{t}-\bP_{Q_{t-1}}\tilde{\bA}\bq_{\perp}^{t}+\bH_{t-1}\balpha^{t}\right) & \approxP\phi'_{n}\left(a\bZ+\bH_{t-1}\balpha^{t,*}\right)\\
 & \approxP\mathbb{E}_{z}\left[\phi'_{n}\left(a\bZ+\bH_{t-1}\balpha^{t,*}\right)\right]\\
 & \approxP\mathbb{E}\left[\phi_{n}\left(\bx^{0},\bZ^{1},\dots,\bZ^{t},a\bZ+\sum_{s=1}^{t}\alpha_{s}^{t,*}\bZ^{s}\right)\right]
\end{align}
where we use Lemma \ref{lem:psedoLipschitz_concentration} in the
second step, and ${\cal H}_{t-1}$$\left(b\right)$ and Lemma \ref{lem:pseudoLipschitz_expectation_wrt_Gaussian}
in the third step. (Here with an abuse of notation, we let $\bZ$ to
be on the same joint space as and independent of $\bZ^{1},\dots,\bZ^{t}$.)
The thesis follows immediately from that 
\begin{equation}
\left(\bZ^{1},\dots,\bZ^{t},a\bZ+\sum_{s=1}^{t}\alpha_{s}^{t,*}\bZ^{s}\right)\ed\left(\bZ^{1},\dots,\bZ^{t},\bZ^{t+1}\right)
\end{equation}
which we now prove.

Let $\tilde{\bZ}=a\bZ+\sum_{s=1}^{t}\alpha_{s}^{t,*}\bZ^{s}$ for brevity.
Observe that $\tilde{\bZ}$ is Gaussian with zero mean and i.i.d. entries,
${\rm Var}\left[\tilde{Z}_{i}\right]$ is a constant independent of
$n$, and $\mathbb{E}\left[\bZ^{s}\tilde{\bZ}^{\sT}\right]=\gamma_{s}\id_{n}$
for some constant $\gamma_{s}$ independent of $n$, for $1\leq s\leq t$.
It suffices to show that ${\rm Var}\left[\tilde{Z}_{i}\right]=\Kappa_{t+1,t+1}$
and $\gamma_{s}=\Kappa_{s,t+1}$. From the above, ${\cal H}_{t}\left(a\right)$
and ${\cal H}_{t-1}\left(b\right)$, we have: 
\begin{equation}
{\rm Var}\left[\tilde{Z}_{i}\right]=\frac{1}{n}\mathbb{E}\left[\left\Vert \tilde{\bZ}\right\Vert _{2}^{2}\right]\approxP\frac{1}{n}\left\Vert \bh^{t+1}\right\Vert _{2}^{2}\approxP\frac{1}{n}\left\Vert \bq^{t}\right\Vert _{2}^{2}\xrightarrow[n\to\infty]{{\rm P}}\Kappa_{t+1,t+1}.
\end{equation}
Similarly, for $s\geq2$, 
\begin{equation}
\gamma_{s}=\frac{1}{n}\mathbb{E}\left[\left\< \bZ^{s},\tilde{\bZ}\right\> \right]\approxP\frac{1}{n}\left\< \bh^{s},\bh^{t+1}\right\> \approxP\frac{1}{n}\left\< \bq^{s-1},\bq^{t}\right\> \xrightarrow[n\to\infty]{{\rm P}}\Kappa_{s,t+1},
\end{equation}
and for $s=1$, 
\begin{equation}
\gamma_{1}=\frac{1}{n}\mathbb{E}\left[\left\< \bZ^{1},\tilde{\bZ}\right\> \right]\approxP\frac{1}{n}\left\< \bh^{1},\bh^{t+1}\right\> \approxP\frac{1}{n}\left\< \bq^{0},\bq^{t}\right\> \xrightarrow[n\to\infty]{{\rm P}}\Kappa_{1,t+1}.
\end{equation}
This completes the proof. 
\end{enumerate}

\subsubsection{Proof of Lemma \ref{lem:new_AMP_approximates_AMP}}

For the recursion (\ref{eq:new_AMP_1})-(\ref{eq:new_AMP_2}), define
the following quantity for each $t\in\mathbb{N}$, 
\begin{equation}
\hat{\bh}^{t+1}=\bA\bq^{t}-\mathsf{b}_{t}\bq^{t-1},\qquad\mathsf{b}_{t}=\E\left[\frac{1}{n}{\rm div}f_{t}\left(\bZ^{t}\right)\right],
\end{equation}
where we take $\hat{\bh}^{1}=\bA\bq^{0}$. 
\begin{lem}
\label{lem:new_AMP_closeTo_hhat}For any $t\in\mathbb{N}_{>0}$, $\frac{1}{\sqrt{n}}\left\Vert \bh^{t+1}-\hat{\bh}^{t+1}\right\Vert _{2}\xrightarrow[n\to\infty]{{\rm P}}0$. 
\end{lem}
\begin{proof}
Denoting the claim as ${\cal H}_{t}$, we prove it by induction. The
base case ${\cal H}_{1}$ is immediate since $\bh^{1}=\hat{\bh}^{1}=\bA\bq^{0}$.
Assuming ${\cal H}_{1},\dots,{\cal H}_{t-1}$, we prove ${\cal H}_{t}$.
Letting $\mathsf{B}_{t}={\rm diag}\left(0,\mathsf{b}_{1},\dots,\mathsf{b}_{t}\right)\in\R^{(t+1)\times(t+1)}$
and $\hat{\bH}_{t-1}=\left[\hat{\bh}^{1}\vert\dots\vert\hat{\bh}^{t}\right]$,
we have $\hat{\bH}_{t-1}=\bA\bQ_{t-1}-\left[0\vert \bQ_{t-2}\right]\mathsf{B}_{t-1}$.
Then since $\bP_{Q_{t-1}}\bq^{t}=\bQ_{t-1}\balpha^{t}$, 
\begin{align}
\bA\bq^{t} & =\bA\bq_{\perp}^{t}+\bA\bQ_{t-1}\balpha^{t}\\
 & =\bA\bq_{\perp}^{t}+\left[0\vert \bQ_{t-2}\right]\mathsf{B}_{t-1}\balpha^{t}+\hat{\bH}_{t-1}\balpha^{t}.
\end{align}
This yields 
\begin{align}
 & \hat{\bh}^{t+1}-\bh^{t+1}=\bP_{Q_{t-1}}\bA\bq_{\perp}^{t}-\mathsf{b}_{t}\bq^{t-1}+\left[0\vert \bQ_{t-2}\right]\mathsf{B}_{t-1}\balpha^{t}+\left(\hat{\bH}_{t-1}-\bH_{t-1}\right)\balpha^{t}\\
 & \quad=\bQ_{t-1}\left(\bQ_{t-1}^{\sT}\bQ_{t-1}\right)^{-1}\bQ_{t-1}^{\sT}A\bq_{\perp}^{t}-\mathsf{b}_{t}\bq^{t-1}+\left[0\vert \bQ_{t-2}\right]\mathsf{B}_{t-1}\balpha^{t}+\left(\hat{\bH}_{t-1}-\bH_{t-1}\right)\balpha^{t}\\
 & \quad\stackrel{\left(a\right)}{=}\bQ_{t-1}\left(\bQ_{t-1}^{\sT}\bQ_{t-1}\right)^{-1}\hat{\bH}_{t-1}^{\sT}\bq_{\perp}^{t}-\mathsf{b}_{t}\bq^{t-1}+\left[0\vert \bQ_{t-2}\right]\mathsf{B}_{t-1}\balpha^{t}+\left(\hat{\bH}_{t-1}-\bH_{t-1}\right)\balpha^{t}\\
 & \quad=\bQ_{t-1}\left(\bQ_{t-1}^{\sT}\bQ_{t-1}\right)^{-1}\bH_{t-1}^{\sT}\bq_{\perp}^{t}-\mathsf{b}_{t}\bq^{t-1}+\left[0\vert \bQ_{t-2}\right]\mathsf{B}_{t-1}\balpha^{t}\\
 & \quad\qquad+\left(\hat{\bH}_{t-1}-\bH_{t-1}\right)\balpha^{t}+\bQ_{t-1}\left(\bQ_{t-1}^{\sT}\bQ_{t-1}\right)^{-1}\left(\hat{\bH}_{t-1}-\bH_{t-1}\right)^{\sT}\bq_{\perp}^{t}\nonumber \\
 & \quad=\sum_{s=1}^{t}c_{s}\bq^{s-1}+\left(\hat{\bH}_{t-1}-\bH_{t-1}\right)\balpha^{t}+\bQ_{t-1}\left(\bQ_{t-1}^{\sT}\bQ_{t-1}\right)^{-1}\left(\hat{\bH}_{t-1}-\bH_{t-1}\right)^{\sT}\bq_{\perp}^{t}
\end{align}
where $\left(a\right)$ holds because $\bQ_{t-1}^{\sT}\bA=\left(\bA\bQ_{t-1}\right)^{\sT}=\hat{\bH}_{t-1}^{\sT}+\mathsf{B}_{t-1}\left[0\vert \bQ_{t-2}\right]^{\sT}$
and $\bQ_{t-2}^{\sT}\bP_{Q_{t-1}}^{\perp}=0$, and 
\begin{equation}
c_{s}=\left[\left(\bQ_{t-1}^{\sT}\bQ_{t-1}\right)^{-1}\bH_{t-1}^{\sT}\bq_{\perp}^{t}\right]_{s}-\mathsf{b}_{s}\left(-\alpha_{s+1}^{t}\right)^{\mathbb{I}_{s\neq t}}.
\end{equation}
By the induction hypothesis, 
\begin{equation}
\frac{1}{\sqrt{n}}\left\Vert \bH_{t-1}-\hat{\bH}_{t-1}\right\Vert _{2}\leq\frac{1}{\sqrt{n}}\left\Vert \bH_{t-1}-\hat{\bH}_{t-1}\right\Vert _{{\rm F}}\xrightarrow[n\to\infty]{{\rm P}}0.
\end{equation}
By Lemma \ref{lem:new_AMP_convergence}, Lemma \ref{lem:pseudoLipschitz_innerProd}
and the non-degeneracy assumption, $\balpha^{t}\xrightarrow[n\to\infty]{{\rm P}}\balpha^{t,*}$
a constant vector in $\mathbb{R}^{t}$. Hence, 
\begin{equation}
\frac{1}{\sqrt{n}}\left\Vert \left(\bH_{t-1}-\hat{\bH}_{t-1}\right)\balpha^{t}\right\Vert _{2}\leq\frac{1}{\sqrt{n}}\left\Vert \bH_{t-1}-\hat{\bH}_{t-1}\right\Vert _{{\rm F}}\left\Vert \balpha^{t,*}\right\Vert _{2}\xrightarrow[n\to\infty]{{\rm P}}0.
\end{equation}
By the non-degeneracy assumption, 
\begin{equation}
\frac{1}{\sqrt{n}}\left\Vert \bQ_{t-1}\left(\bQ_{t-1}^{\sT}\bQ_{t-1}\right)^{-1}\left(\hat{\bH}_{t-1}-\bH_{t-1}\right)^{\sT}\bq_{\perp}^{t}\right\Vert _{2}\leq\frac{1}{\sqrt{n}}\left\Vert \hat{\bH}_{t-1}-\bH_{t-1}\right\Vert _{{\rm F}}\frac{1}{c_{t}\sqrt{n}}\left\Vert \bq^{t}\right\Vert _{2}\xrightarrow[n\to\infty]{{\rm P}}0,
\end{equation}
where $\frac{1}{\sqrt{n}}\left\Vert \bq^{t}\right\Vert _{2}$ converges
in probability to a finite constant by Lemma \ref{lem:new_AMP_convergence}.
We claim that $\frac{1}{\sqrt{n}}c_{s}\left\Vert \bq^{s-1}\right\Vert _{2}\xrightarrow[n\to\infty]{{\rm P}}0$
for $s=1,\dots,t$. Then the thesis follows from this claim.

To prove the claim, denoting $R=\frac{1}{n}\left(\bQ_{t-1}^{\sT}\bQ_{t-1}\right)^{-1}$
for brevity, we note that 
\begin{equation}
c_{s}=\sum_{r=1}^{t}R_{s,r}\frac{1}{n}\left\< \bh^{r},\bq^{t}-\sum_{\ell=1}^{t}\alpha_{\ell}^{t}\bq^{\ell-1}\right\> -\mathsf{b}_{s}\left(-\alpha_{s+1}^{t}\right)^{\mathbb{I}_{s\neq t}}.
\end{equation}
We now analyze $c_{s}$. By Lemma \ref{lem:new_AMP_convergence},
\begin{equation}
\frac{1}{n}\left\< \bh^{r},\bq^{0}\right\> \approxP\mathbb{E}\left[\frac{1}{n}\left\< \bZ^{r},f_{0}\left(\bx^{0}\right)\right\> \right]=0
\end{equation}
since $\bZ^{r}$ has zero mean. By Lemmas \ref{lem:new_AMP_convergence}
and \ref{lem:Stein}, for $j=2,\dots,t-1$, 
\begin{align}
\frac{1}{n}\left\< \bh^{r},\bq^{j}\right\>  & \approxP\mathbb{E}\left[\frac{1}{n}\left\< \bZ^{r},f_{j}\left(\bZ^{j}\right)\right\> \right]\\
 & =\Kappa_{r,j}\mathbb{E}\left[\frac{1}{n}{\rm div}f_{j}\left(\bZ^{j}\right)\right]\\
 & \approxP\frac{1}{n}\left\< \bq^{r-1},\bq^{j-1}\right\> \mathsf{b}_{j}.
\end{align}
Therefore, 
\begin{equation}
c_{s}\approxP\left\{ \sum_{r=1}^{t}R_{s,r}\frac{1}{n}\left\< \bq^{r-1},\mathsf{b}_{t}\bq^{t-1}-\sum_{\ell=2}^{t}\alpha_{1}^{t}\mathsf{b}_{\ell-1}\bq^{\ell-2}\right\> -\mathsf{b}_{s}\left(-\alpha_{s+1}^{t}\right)^{\mathbb{I}_{s\neq t}}\right\} .
\end{equation}
Identifying $\frac{1}{n}\left\< \bq^{r-1},\bq^{j-1}\right\> =\left(R^{-1}\right)_{r,j}$,
we get 
\begin{equation}
c_{s}\approxP\left\{ \mathsf{b}_{t}\mathbb{I}_{t=s}-\sum_{\ell=2}^{t}\alpha_{\ell}^{t}\mathsf{b}_{\ell-1}\mathbb{I}_{\ell-1=s}-\mathsf{b}_{s}\left(-\alpha_{s+1}^{t}\right)^{\mathbb{I}_{s\neq t}}\right\} 
\end{equation}
i.e. $c_{s}\xrightarrow[n\to\infty]{{\rm P}}0$. Finally, since $\frac{1}{\sqrt{n}}\left\Vert \bq^{s-1}\right\Vert _{2}$
converges in probability to a finite constant by Lemma \ref{lem:new_AMP_convergence},
the claim is proven. 
\end{proof}
\begin{proof}[Proof of Lemma \ref{lem:new_AMP_approximates_AMP}]
Let ${\cal H}_{t}$ be the statement $\frac{1}{\sqrt{n}}\left\Vert \bq^{t}-\bm^{t}\right\Vert _{2}\xrightarrow[n\to\infty]{{\rm P}}0$
and $\frac{1}{\sqrt{n}}\left\Vert \bh^{t+1}-\bx^{t+1}\right\Vert _{2}\xrightarrow[n\to\infty]{{\rm P}}0$.
We prove it by induction. The base case ${\cal H}_{0}$ is trivial
because $\bq^{0}=\bm^{0}$ and $\bh^{1}=\bx^{1}$.

We now assume ${\cal H}_{t-1}$ is true and we show ${\cal H}_{t}$.
We have: 
\begin{equation}
\frac{1}{\sqrt{n}}\left\Vert \bq^{t}-\bm^{t}\right\Vert _{2}=\frac{1}{\sqrt{n}}\left\Vert f_{t}\left(\bh^{t}\right)-f_{t}\left(\bx^{t}\right)\right\Vert _{2}\leq L_{t}\frac{1}{\sqrt{n}}\left\Vert \bh^{t}-\bx^{t}\right\Vert _{2}\xrightarrow[n\to\infty]{{\rm P}}0,
\end{equation}
using that $f_{t}$ is uniformly Lipschitz and the induction hypothesis
${\cal H}_{t-1}$. Further, we will prove that $\frac{1}{\sqrt{n}}\left\Vert \hat{\bh}^{t+1}-\bx^{t+1}\right\Vert _{2}\xrightarrow[n\to\infty]{{\rm P}}0$,
which together with Lemma \ref{lem:new_AMP_closeTo_hhat} yields ${\cal H}_{t}$.
We have: 
\begin{equation}
\hat{\bh}^{t+1}-\bx^{t+1}=A(\bq^{t}-\bm^{t})-\mathsf{b}_{t}(\bq^{t-1}-\bm^{t-1}),
\end{equation}
thus by Theorem \ref{thm:A_norm} and ${\cal H}_{t-1}$, 
\begin{equation}
\frac{1}{\sqrt{n}}\left\Vert \hat{\bh}^{t+1}-\bx^{t+1}\right\Vert _{2}\leq\left\Vert \bA\right\Vert _{{\rm op}}\frac{1}{\sqrt{n}}\left\Vert \bq^{t}-\bm^{t}\right\Vert _{2}+\mathsf{b}_{t}\frac{1}{\sqrt{n}}\left\Vert \bq^{t-1}-\bm^{t-1}\right\Vert _{2}\xrightarrow[n\to\infty]{{\rm P}}0.
\end{equation}
This concludes the induction. 
\end{proof}

\subsubsection{Proof of Lemma \ref{lem:SE_perturbated}}

Let us first check assumption \ref{assumption:gaussian-conv} for
the perturbed setting $\left\{ \bx^{0},f_{t}^{\epsilon y}\right\} $.
Consider $s,t\geq1$, $\Kappa$ a $2\times2$ covariance matrix and
$(\bZ^{s},\bZ^{t})\in(\R^{n})^{2}$, $(\bZ^{s},\bZ^{t})\sim\normal\left(0,\Kappa\otimes \id_{n}\right)$.
Note that $\Kappa$ is deterministic, not depending on the perturbation
$\by$. We denote the expectation over $(\bZ^{s},\bZ^{t})$ as $\E_{Z}$.
We have: 
\begin{align}
\E_{Z}\left[\frac{1}{n}\left\< f_{s}^{\epsilon y}(\bZ^{s}),f_{t}^{\epsilon y}(\bZ^{t})\right\> \right] & =\E_{Z}\left[\frac{1}{n}\left\< f_{s}(\bZ^{s}),f_{t}(\bZ^{t})\right\> \right]+\epsilon\E_{Z}\left[\frac{1}{n}\left\< f_{s}(\bZ^{s}),\by^{t}\right\> \right]\\
 & \quad+\epsilon\E_{Z}\left[\frac{1}{n}\left\< \by^{s},f_{t}(\bZ^{t})\right\> \right]+\epsilon^{2}\frac{1}{n}\left\< \by^{s},\by^{t}\right\> \\
 & =\E_{Z}\left[\frac{1}{n}\left\< f_{s}(\bZ^{s}),f_{t}(\bZ^{t})\right\> \right]+\epsilon\frac{1}{n}\left\< \E_{Z}\left[f_{s}(\bZ^{s})\right],\by^{t}\right\> \\
 & \quad+\epsilon\frac{1}{n}\left\< \by^{s},\E_{Z}\left[f_{t}(\bZ^{t})\right]\right\> +\epsilon^{2}\frac{1}{n}\left\< \by^{s},\by^{t}\right\> .
\end{align}

\begin{itemize}
\item The first term does not depend on the perturbation and is thus deterministic.
By assumption \ref{assumption:gaussian-conv} for the setting $\left\{ \bx^{0},f_{t}\right\} $,
$\E_{Z}\left[\frac{1}{n}\left\< f_{s}(\bZ^{s}),f_{t}(\bZ^{t})\right\> \right]$
converges to a (deterministic) limit. 
\item The second term is Gaussian, with mean zero and variance 
\begin{equation}
\frac{1}{n^{2}}\left\Vert \E\left[f_{s}(\bZ^{s})\right]\right\Vert _{2}^{2}\leq\frac{1}{n^{2}}\E\left[\left\Vert f_{s}(\bZ^{s})\right\Vert _{2}^{2}\right]\leq\frac{C}{n},
\end{equation}
for $C$ a constant large enough, using again the assumption \ref{assumption:gaussian-conv}
for the setting $\left\{ \bx^{0},f_{t}\right\} $. Thus, $\frac{1}{n}\left\< \E_{Z}\left[f_{s}(\bZ^{s})\right],\by^{t}\right\> $
is a Gaussian random variable, of standard deviation smaller than
$\sqrt{C}/\sqrt{n}$. Then if $\eta>0$, 
\begin{equation}
\Pr\left(\left\vert \frac{1}{n}\left\< \E_{Z}\left[f_{s}(\bZ^{s})\right],\by^{t}\right\> \right\vert \geq\eta\right)\leq\Pr\left(\left\vert \normal(0,1)\right\vert \geq\eta\frac{\sqrt{n}}{\sqrt{C}}\right)\leq\frac{\sqrt{C}}{\eta\sqrt{n}}\frac{1}{\sqrt{2\pi}}\exp\left(-\frac{1}{2}\frac{\eta^{2}n}{C}\right)
\end{equation}
which is summable. Using Borel-Cantelli's lemma, it is then easy to
show that 
\begin{equation}
\frac{1}{n}\left\< \E_{Z}\left[f_{s}(\bZ^{s})\right],\by^{t}\right\> \xrightarrow[n\to\infty]{a.s}0.
\end{equation}
\item The treatment of the third term is the same as for the second term. 
\item Using the law of large numbers, we get that 
\begin{equation}
\frac{1}{n}\left\< \by^{s},\by^{t}\right\> \xrightarrow[n\to\infty]{a.s.}\mathbbm{1}_{s=t}.
\end{equation}
\end{itemize}
Putting things together, we get almost surely 
\begin{equation}
\lim_{n\to\infty}\E_{Z}\left[\frac{1}{n}\left\< f_{s}^{\epsilon y}(\bZ^{s}),f_{t}^{\epsilon y}(\bZ^{t})\right\> \right]=\lim_{n\to\infty}\E_{Z}\left[\frac{1}{n}\left\< f_{s}(\bZ^{s}),f_{t}(\bZ^{t})\right\> \right]+\epsilon^{2}\mathbbm{1}_{s=t}.\label{eq:gaussian-conv-perturbed}
\end{equation}

The proof of assumptions \ref{assumption:x^0-conv}, \ref{assumption:mixed-conv}
are very similar, here we only state the resulting expressions: almost
surely, 
\begin{align}
\lim_{n\rightarrow\infty}\frac{1}{n}\left\< f_{0}^{\epsilon y}(\bx^{0}),f_{0}^{\epsilon y}(\bx^{0})\right\>  & =\lim_{n\rightarrow\infty}\frac{1}{n}\left\< f_{0}(\bx^{0}),f_{0}(\bx^{0})\right\> +\epsilon^{2},\label{eq:x0-conv-perturbed}\\
\lim_{n\rightarrow\infty}\E\left[\frac{1}{n}\left\< f_{0}^{\epsilon y}(\bx^{0}),f_{t}^{\epsilon y}(\bZ^{t})\right\> \right] & =\lim_{n\rightarrow\infty}\E\left[\frac{1}{n}\left\< f_{0}(\bx^{0}),f_{t}(\bZ^{t})\right\> \right].\label{eq:mixed-conv-perturbed}
\end{align}
Using equations (\ref{eq:gaussian-conv-perturbed}), (\ref{eq:x0-conv-perturbed}),
(\ref{eq:mixed-conv-perturbed}), it is a simple induction that the
state evolution for the perturbed setting $\left\{ \bx^{0},f_{t}^{\epsilon y}\right\} $
is indeed non-random almost surely.

\subsubsection{Proof of Lemma \ref{lem:good-cond-perturbed}}

By definition, 
\begin{equation}
\bq_{\perp}^{\epsilon y,t}=\bP_{Q_{t-1}}^{\perp}f_{t}(\bh^{\epsilon y,t})+\epsilon \bP_{Q_{t-1}}^{\perp}\by^{t}.
\end{equation}
If we denote $\mathcal{F}_{t}$ as the $\sigma$-algebra generated
by $\bh^{\epsilon y,1},\dots,\bh^{\epsilon y,t},\by^{1},\dots,\by^{t-1}$, it
follows that 
\begin{equation}
\bq_{\perp}^{\epsilon y,t}\vert_{\mathcal{F}_{t}}\sim\normal\left(\bP_{Q_{t-1}}^{\perp}f_{t}(\bh^{\epsilon y,t}),\epsilon^{2}\bP_{Q_{t-1}}^{\perp}\right).
\end{equation}
When $n>t$, this conditional distribution is almost surely non-zero.
Thus when $n\geq t$, the matrix $\bQ_{t-1}$ has full column rank.

To lower bound the minimum singular value of $\bQ_{t-1}$, a more careful
treatment is required. Using \cite[Lemma 8]{Bayati2011}, it is sufficient
to check that there exists a constant $c_{\epsilon}$ such that almost
surely, for $n$ sufficiently large, 
\begin{equation}
\frac{1}{n}\left\Vert \bq_{\perp}^{\epsilon y,t}\right\Vert ^{2}\geq c_{\epsilon}.
\end{equation}
We have 
\begin{align}
\Pr\left(\frac{1}{n}\left\Vert \bq_{\perp}^{\epsilon y,t}\right\Vert ^{2}\leq c_{\epsilon}\middle|\mathcal{F}_{t}\right) & =\Pr\left(\left\Vert \normal\left(\bP_{Q_{t-1}}^{\perp}f_{t}(\bh^{\epsilon y,t}),\epsilon^{2}\bP_{Q_{t-1}}^{\perp}\right)\right\Vert ^{2}\leq c_{\epsilon}n\bigg\vert\mathcal{F}_{t}\right)\\
 & \leq\Pr\left(\left\Vert \normal\left(0,\epsilon^{2}\bP_{Q_{t-1}}^{\perp}\right)\right\Vert ^{2}\leq c_{\epsilon}n\bigg\vert\mathcal{F}_{t}\right)\\
 & =\Pr\left(\chi_{n-t}\leq\frac{c_{\epsilon}n}{\epsilon^{2}}\right)\\
 & =\Pr\left(\frac{\chi_{n-t}}{n-t}\leq\frac{c_{\epsilon}}{\epsilon^{2}}\frac{n}{n-t}\right).
\end{align}
We can choose $c_{\epsilon}$ such that $c_{\epsilon}/\epsilon^{2}=1/4$,
and consider only the case $n\geq2t$, so that $n/(n-t)\leq2$. We
then get: 
\begin{equation}
\Pr\left(\frac{1}{n}\left\Vert \bq_{\perp}^{\epsilon y,t}\right\Vert ^{2}\leq c_{\epsilon}\bigg\vert\mathcal{F}_{t}\right)\leq\Pr\left(\frac{\chi_{n-t}}{n-t}\leq\frac{1}{2}\right).
\end{equation}
Using concentration of the chi-squared variable, it is easy to show
that $\Pr\left(\frac{\chi_{n-t}}{n-t}\leq\frac{1}{2}\right)$ is summable
over $n$. Taking expectation of the last inequality, we get 
\begin{equation}
\sum_{n}\Pr\left(\frac{1}{n}\left\Vert \bq_{\perp}^{\epsilon y,t}\right\Vert ^{2}\leq c_{\epsilon}\right)<+\infty.
\end{equation}
Then Borel-Cantelli's lemma concludes the proof.

\subsubsection{Proof of Lemma \ref{lem:pseudo-lip_and_gaussians}}

Define $k$ as the order of the sequence $\left\{ \phi_{n}\right\} $
of uniformly pseudo-Lipschitz functions, and $L$ as its pseudo-Lipschitz
constant. Under any coupling of $\bZ$ and $\tilde{\bZ}$, 
\begin{align}
\left\vert \E\left[\phi_{n}\left(\bZ\right)\right]-\E\left[\phi_{n}\left(\tilde{\bZ}\right)\right]\right\vert  & \leq L\E\left[\left(1+\left(\frac{\left\Vert \bZ\right\Vert _{2}}{\sqrt{n}}\right)^{k-1}+\left(\frac{\left\Vert \tilde{\bZ}\right\Vert _{2}}{\sqrt{n}}\right)^{k-1}\right)\frac{\left\Vert \bZ-\tilde{\bZ}\right\Vert _{2}}{\sqrt{n}}\right]\\
 & \leq L\E\left[\left(1+\left(\frac{\left\Vert \bZ\right\Vert _{2}}{\sqrt{n}}\right)^{k-1}+\left(\frac{\left\Vert \tilde{\bZ}\right\Vert _{2}}{\sqrt{n}}\right)^{k-1}\right)^{2}\right]^{1/2}\frac{1}{\sqrt{n}}\E\left[\left\Vert \bZ-\tilde{\bZ}\right\Vert _{2}^{2}\right]^{1/2}.
\end{align}
Taking the infimum over all possible coupling of $\bZ\sim\normal\left(0,\Kappa\otimes \id_{n}\right)$
and $\tilde{\bZ}\sim\normal\left(0,\tilde{\Kappa}\otimes \id_{n}\right)$,
one gets a bound involving the Wasserstein distance $W_{2}$: 
\begin{align}
 & \left\vert \E\left[\phi_{n}\left(\bZ\right)\right]-\E\left[\phi_{n}\left(\tilde{\bZ}\right)\right]\right\vert \\
 & \leq\sqrt{3}L\left(1+\frac{\E\left[\left\Vert \bZ\right\Vert _{2}^{2(k-1)}\right]}{n^{k-1}}+\frac{\E\left[\left\Vert \tilde{\bZ}\right\Vert _{2}^{2(k-1)}\right]}{n^{k-1}}\right)^{1/2}\frac{1}{\sqrt{n}}W_{2}\left(\normal\left(0,\Kappa\otimes \id_{n}\right),\normal\left(0,\tilde{\Kappa}\otimes \id_{n}\right)\right).\label{eq:pseudo-lip_and_gaussians_aux}
\end{align}
We then use the two following identities for the Wasserstein distance:
\begin{align}
W_{2}(\mu\otimes\nu,\mu'\otimes\nu')^{2} & =W_{2}(\mu,\mu')^{2}+W_{2}(\nu,\nu')^{2},\\
W_{2}\left(\normal\left(0,\Kappa\right),\normal\left(0,\tilde{\Kappa}\right)\right)^{2} & ={\rm Tr}\left(\Kappa+\tilde{\Kappa}-2\left(\Kappa^{1/2}\tilde{\Kappa}\Kappa^{1/2}\right)^{1/2}\right).
\end{align}
For a proof of the second identity, see \cite[Proposition 7]{givens1984}.
It follows that 
\begin{equation}
W_{2}\left(\normal\left(0,\Kappa\otimes \id_{n}\right),\normal\left(0,\tilde{\Kappa}\otimes \id_{n}\right)\right)^{2}=n{\rm Tr}\left(\Kappa+\tilde{\Kappa}-2\left(\Kappa^{1/2}\tilde{\Kappa}\Kappa^{1/2}\right)^{1/2}\right).
\end{equation}
Moreover, $Z\ed\left(\Kappa^{1/2}\otimes \id_{n}\right)X$
where $X\sim\normal\left(0,\id_{nt}\right)$. Thus: 
\begin{equation}
\E\left[\left\Vert Z\right\Vert _{2}^{2(k-1)}\right]\leq\left\Vert \Kappa^{1/2}\otimes \id_{n}\right\Vert _{2}^{2(k-1)}\E\left[\left\Vert X\right\Vert _{2}^{2(k-1)}\right]=\left\Vert \Kappa\right\Vert _{{\rm op}}^{k-1}\E\left[\left(\chi_{nt}^{2}\right)^{k-1}\right].
\end{equation}
Using expressions for moments of chi-square variables, we get: 
\begin{equation}
\E\left[\left(\chi_{n\times t}^{2}\right)^{k-1}\right]=nt(nt+2)\dots(nt+2(k-2))\leq n^{k-1}t^{k-1}(1+2(k-2))^{k-1}=C(k,t)n^{k-1}
\end{equation}
for a constant $C(k,t)$ that depends only on $k$ and $t$. Back
to inequality (\ref{eq:pseudo-lip_and_gaussians_aux}), 
\begin{align}
 & \left\vert \E\left[\phi_{n}\left(\bZ\right)\right]-\E\left[\phi_{n}\left(\tilde{\bZ}\right)\right]\right\vert \\
 & \leq\sqrt{3}L\left(1+C(k,t)\left(\left\Vert \Kappa\right\Vert _{{\rm op}}^{k-1}+\left\Vert \tilde{\Kappa}\right\Vert _{{\rm op}}^{k-1}\right)\right)^{1/2}\left({\rm Tr}\left(\Kappa+\tilde{\Kappa}-2\left(\Kappa^{1/2}\tilde{\Kappa}\Kappa^{1/2}\right)^{1/2}\right)\right)^{1/2}.
\end{align}
Notice that this bound is independent of $n$, and converges to 0
as $\tilde{\Kappa}\to\Kappa$.

\subsubsection{Proof of Lemma \ref{lem:SE_eps_to_zero}}

This lemma will be shown by induction.

\paragraph{Initialization.}

According to (\ref{eq:x0-conv-perturbed}), 
\begin{equation}
\Kappa_{1,1}^{\epsilon}=\Kappa_{1,1}+\epsilon^{2}\xrightarrow[\epsilon\to0]{}\Kappa_{1,1}.
\end{equation}

\paragraph{Induction.}

Let $t$ be a non-negative integer. Assume that by the induction hypothesis,
for any $r,s\leq t$, $\Kappa_{r,s}^{\epsilon}\xrightarrow{}\Kappa_{r,s}$.
Then: 
\begin{equation}
\Kappa_{s+1,t+1}^{\epsilon}\stackrel{{\rm a.s.}}{=}\lim_{n\to\infty}\mathbb{E}\left[\frac{1}{n}\left\< f_{s}^{\epsilon y}\left(\bZ^{\epsilon,s}\right),f_{t}^{\epsilon y}\left(\bZ^{\epsilon,t}\right)\right\> \right],
\end{equation}
where $(\bZ^{\epsilon,s},\bZ^{\epsilon,t})$ is a Gaussian vector, whose
covariance is determined by $\Kappa_{s,s}^{\epsilon}$, $\Kappa_{t,t}^{\epsilon}$
and $\Kappa_{s,t}^{\epsilon}$. Using (\ref{eq:gaussian-conv-perturbed}),
we have 
\begin{equation}
\Kappa_{s+1,t+1}^{\epsilon}\stackrel{{\rm a.s.}}{=}\lim_{n\to\infty}\mathbb{E}\left[\frac{1}{n}\left\< f_{s}\left(\bZ^{\epsilon,s}\right),f_{t}\left(\bZ^{\epsilon,t}\right)\right\> \right]+\epsilon^{2}\mathbbm{1}_{s=t}.
\end{equation}
The sequence of functions $(z^{s},z^{t})\mapsto\frac{1}{n}\left\< f_{s}\left(z^{s}\right),f_{t}\left(z^{t}\right)\right\> $
is uniformly pseudo-Lipschitz by Lemma \ref{lem:pseudoLipschitz_innerProd},
thus Lemma \ref{lem:pseudo-lip_and_gaussians} and the induction hypothesis
jointly ensure that 
\begin{equation}
\lim_{\epsilon\to0}\lim_{n\to\infty}\mathbb{E}\left[\frac{1}{n}\left\< f_{s}\left(\bZ^{\epsilon,s}\right),f_{t}\left(\bZ^{\epsilon,t}\right)\right\> \right]=\lim_{n\to\infty}\mathbb{E}\left[\frac{1}{n}\left\< f_{s}\left(\bZ^{s}\right),f_{t}\left(\bZ^{t}\right)\right\> \right]=\Kappa_{s+1,t+1},
\end{equation}
where $(\bZ^{s},\bZ^{t})\in(\R^{n})^{2}$, $(\bZ^{s},\bZ^{t})\sim\normal\left(0,\Kappa\otimes \id_{n}\right)$.
Thus, we indeed get 
\begin{equation}
\Kappa_{s+1,t+1}^{\epsilon}\xrightarrow[\epsilon\to0]{}\Kappa_{s+1,t+1}.
\end{equation}

To finish the induction reasoning, one can check similarly that $\Kappa_{1,t+1}^{\epsilon}\xrightarrow[\epsilon\to0]{}\Kappa_{1,t+1}$.

\subsubsection{Proof of Lemma \ref{lem:AMP_eps_to_0}}
\label{sec:AMP_eps_to_0}

First, it is easy to check by induction that there exist constants
$\tilde{C}_{t}$, $\tilde{C}'_{t}$ and $\tilde{C}''_{t}$ independent
of $n$ such that for all $\epsilon\leq1$, w.h.p. 
\begin{align}
\frac{1}{\sqrt{n}}\left\Vert \bm^{\epsilon y,t}\right\Vert _{2} & \leq\tilde{C}'_{t},\label{eq:m^t_size}\\
\frac{1}{\sqrt{n}}\left\Vert \bx^{\epsilon y,t+1}\right\Vert _{2} & \leq\tilde{C}_{t}.
\end{align}
Indeed, one only needs to use that the functions involved are uniformly
Lipschitz and Theorem \ref{thm:A_norm}. Note that these inequalities
hold for the original AMP iterates by taking $\epsilon=0$.

We now prove our lemma by induction.

\paragraph{Initialization.}

We have 
\begin{equation}
\frac{1}{\sqrt{n}}\left\Vert \bm^{\epsilon y,0}-\bm^{0}\right\Vert _{2}=\epsilon\frac{\left\Vert \by^{0}\right\Vert _{2}}{\sqrt{n}}\leq2\epsilon\quad{\rm w.h.p.},
\end{equation}
by the law of large numbers. Thus we choose $h_{0}'(\epsilon)=2\epsilon$.
Furthermore, 
\begin{equation}
\frac{1}{\sqrt{n}}\left\Vert \bx^{\epsilon y,1}-\bx^{1}\right\Vert _{2}\leq\left\Vert \bA\right\Vert _{{\rm op}}\frac{1}{\sqrt{n}}\left\Vert \bm^{\epsilon y,0}-\bm^{0}\right\Vert _{2}\leq6\epsilon\quad{\rm w.h.p.},
\end{equation}
by Theorem \ref{thm:A_norm}. Thus we choose $h_{0}(\epsilon)=6\epsilon$.

\paragraph{Induction.}

We assume here that $\Kappa_{1,1},\dots,\Kappa_{t,t}>0$. By induction
hypothesis, we have already defined $h_{0}(\epsilon),h_{0}'(\epsilon),\dots,h_{t-1}(\epsilon),h_{t-1}'(\epsilon)$.
We now choose $h_{t}(\epsilon)$ and $h_{t}'(\epsilon)$. We have
\begin{align}
\frac{1}{\sqrt{n}}\left\Vert \bm^{\epsilon y,t}-\bm^{t}\right\Vert _{2} & \leq\frac{1}{\sqrt{n}}\left\Vert f_{t}\left(\bx^{\epsilon y,t}\right)-f_{t}\left(\bx^{t}\right)+\epsilon \by^{t}\right\Vert _{2}\\
 & \leq L_{t}\frac{1}{\sqrt{n}}\left\Vert \bx^{\epsilon y,t}-\bx^{t}\right\Vert _{2}+\epsilon\frac{\left\Vert \by^{t}\right\Vert _{2}}{\sqrt{n}}\leq L_{t}h_{t-1}(\epsilon)+2\epsilon\quad{\rm w.h.p.}
\end{align}
using that $f_{t}$ is uniformly Lipschitz with Lipschitz constant
$L_{t}$. Thus we choose $h_{t}'(\epsilon)=L_{t}h_{t-1}(\epsilon)+2\epsilon$,
which converges to zero as $\epsilon\to0$. Furthermore, 
\begin{align}
\frac{1}{\sqrt{n}}\left\Vert \bx^{\epsilon y,t+1}-\bx^{t+1}\right\Vert _{2} & \leq\left\Vert \bA\right\Vert _{{\rm op}}\frac{1}{\sqrt{n}}\left\Vert \bm^{\epsilon y,t}-\bm^{t}\right\Vert _{2}+\frac{1}{\sqrt{n}}\left\Vert b_{t}^{\epsilon y}\bm^{\epsilon y,t-1}-b_{t}\bm^{t-1}\right\Vert _{2}\\
 & \leq3h_{t}'(\epsilon)+\frac{1}{\sqrt{n}}\left\Vert b_{t}^{\epsilon y}\bm^{\epsilon y,t-1}-b_{t}\bm^{t-1}\right\Vert _{2}\quad{\rm w.h.p.}
\end{align}
by Theorem \ref{thm:A_norm}. We have from (\ref{eq:m^t_size}): 
\begin{align}
\frac{1}{\sqrt{n}}\left\Vert \mathsf{b}_{t}^{\epsilon y}\bm^{\epsilon y,t-1}-\mathsf{b}_{t}\bm^{t-1}\right\Vert _{2} & \leq\left\vert \mathsf{b}_{t}^{\epsilon y}\right\vert \frac{1}{\sqrt{n}}\left\Vert \bm^{\epsilon y,t-1}-\bm^{t-1}\right\Vert _{2}+\left\vert \mathsf{b}_{t}^{\epsilon y}-\mathsf{b}_{t}\right\vert \frac{1}{\sqrt{n}}\left\Vert \bm^{t-1}\right\Vert _{2}\\
 & \leq Lh'_{t-1}(\epsilon)+\left\vert \mathsf{b}_{t}^{\epsilon y}-\mathsf{b}_{t}\right\vert \tilde{C}'_{t-1}.\label{eq:onsager_diff_bound}
\end{align}
Since $\Kappa_{t,t}^{\epsilon}\rightarrow\Kappa_{t,t}$ when $\epsilon\to0$
from Lemma \ref{lem:SE_eps_to_zero} and $\Kappa_{t,t}>0$, we have
$\Kappa_{t,t}^{\epsilon}>0$ for sufficiently small $\epsilon$. Then
using Lemma \ref{lem:Stein}, with $\bZ\sim\normal\left(0,\id_{n}\right)$,
we get 
\begin{align}
\left\vert \mathsf{b}_{t}^{\epsilon y}-\mathsf{b}_{t}\right\vert  & =\left\vert \E\left[\frac{1}{n}{\rm div}f_{t}\left(\sqrt{\Kappa_{t,t}^{\epsilon}}\bZ\right)\right]-\E\left[\frac{1}{n}{\rm div}f_{t}\left(\sqrt{\Kappa_{t,t}}\bZ\right)\right]\right\vert \\
 & =\left\vert \frac{1}{\sqrt{\Kappa_{t,t}^{\epsilon}}}\E\left[\frac{1}{n}\left\< \bZ,f_{t}\left(\sqrt{\Kappa_{t,t}^{\epsilon}}\bZ\right)\right\> \right]-\frac{1}{\sqrt{\Kappa_{t,t}}}\E\left[\frac{1}{n}\left\< \bZ,f_{t}\left(\sqrt{\Kappa_{t,t}}\bZ\right)\right\> \right]\right\vert \\
 & \leq\frac{1}{\sqrt{\Kappa_{t,t}^{\epsilon}}}\left\vert \E\left[\frac{1}{n}\left\< \bZ,f_{t}\left(\sqrt{\Kappa_{t,t}^{\epsilon}}\bZ\right)-f_{t}\left(\sqrt{\Kappa_{t,t}}\bZ\right)\right\> \right]\right\vert +\left\vert \frac{1}{\sqrt{\Kappa_{t,t}^{\epsilon}}}-\frac{1}{\sqrt{\Kappa_{t,t}}}\right\vert \E\left[\frac{1}{n}\left\vert \left\< \bZ,f_{t}\left(\sqrt{\Kappa_{t,t}}\bZ\right)\right\> \right\vert \right]\\
 & \leq\frac{1}{\sqrt{\Kappa_{t,t}^{\epsilon}}}\E\left[\frac{1}{n}\left\Vert \bZ\right\Vert _{2}^{2}\right]^{1/2}\E\left[\frac{1}{n}\left\Vert f_{t}\left(\sqrt{\Kappa_{t,t}^{\epsilon}}\bZ\right)-f_{t}\left(\sqrt{\Kappa_{t,t}}\bZ\right)\right\Vert _{2}^{2}\right]^{1/2}\\
 & \quad+\left\vert \frac{1}{\sqrt{\Kappa_{t,t}^{\epsilon}}}-\frac{1}{\sqrt{\Kappa_{t,t}}}\right\vert \E\left[\left\vert \frac{1}{n}\left\< \bZ,f_{t}(0)\right\> \right\vert +\left\vert \frac{1}{n}\left\< \bZ,f_{t}\left(\sqrt{\Kappa_{t,t}}\bZ\right)-f_{t}(0)\right\> \right\vert \right]\\
 & \leq\frac{1}{\sqrt{\Kappa_{t,t}^{\epsilon}}}\E\left[\frac{1}{n}\left\Vert \bZ\right\Vert _{2}^{2}\right]L_{t}\left(\sqrt{\Kappa_{t,t}^{\epsilon}}-\sqrt{\Kappa_{t,t}}\right)+\left\vert \frac{1}{\sqrt{\Kappa_{t,t}^{\epsilon}}}-\frac{1}{\sqrt{\Kappa_{t,t}}}\right\vert \left(\frac{\left\Vert f_{t}(0)\right\Vert _{2}^{2}}{n}+\E\left[\frac{1}{n}\left\Vert \bZ\right\Vert _{2}^{2}\right]L_{t}\sqrt{\Kappa_{t,t}}\right)\\
 & \leq\frac{1}{\sqrt{\Kappa_{t,t}^{\epsilon}}}L_{t}\left(\sqrt{\Kappa_{t,t}^{\epsilon}}-\sqrt{\Kappa_{t,t}}\right)+\left\vert \frac{1}{\sqrt{\Kappa_{t,t}^{\epsilon}}}-\frac{1}{\sqrt{\Kappa_{t,t}}}\right\vert \left(\frac{\left\Vert f_{t}(0)\right\Vert _{2}^{2}}{\sqrt{n}}+L_{t}\sqrt{\Kappa_{t,t}}\right)\label{eq:b_t_diff_bound}
\end{align}
Since the quantity $\left\Vert f_{t}(0)\right\Vert _{2}^{2}/n$ is
upper bounded by a constant independent of $n$, we can plug (\ref{eq:b_t_diff_bound})
into (\ref{eq:onsager_diff_bound}), and choose correspondingly a
function $h_{t}(\epsilon)$ such that $h_{t}(\epsilon)\rightarrow0$
when $\epsilon\to0$, enabled by the fact $\Kappa_{t,t}^{\epsilon}\rightarrow\Kappa_{t,t}$.

\section{Proof of Theorem \ref{thm:AMP_convergence_asym} and Corollary \ref{coro:Empirical} (Asymmetric AMP)}
\label{sec:Asymmetric-AMP}

\begin{proof}[Proof of Theorem \ref{thm:AMP_convergence_asym}]
We reduce this case to the asymmetric case, as in \cite{Javanmard2013}.
Consider 
\[
\bA_{s}=\sqrt{\frac{\delta}{\delta+1}}\left[\begin{array}{cc}
\bB & \bA\\
\bA^{\sT} & \bC
\end{array}\right]\quad{\rm and}\quad \bx^{0}=\left[\begin{array}{c}
0\\
\bu^{0}
\end{array}\right].
\]
where $\bB\sim \GOE\left(m\right)$ and $\sqrt{\delta}\bC\sim{\GOE}\left(n\right)$
are independent of each other and of $\bA$. It is easy to see that
$\bA_{s}\sim{\GOE}\left(N\right)$, where $N=m+n$. We further let
$f_{t}:\,\mathbb{R}^{N}\to\mathbb{R}^{N}$ be such that 
\[
f_{2t+1}\left(\bx\right)=\sqrt{\frac{\delta+1}{\delta}}\left[\begin{array}{c}
g_{t}\left(x_{1},\dots,x_{m}\right)\\
0
\end{array}\right],\quad f_{2t}\left(\bx\right)=\sqrt{\frac{\delta+1}{\delta}}\left[\begin{array}{c}
0\\
e_{t}\left(x_{m+1},\dots,x_{N}\right)
\end{array}\right]
\]
for any $\bx\in\mathbb{R}^{N}$. We can define the symmetric AMP recursion
$\left\{ \bx^{t},\bm^{t}\vert f_{t},\bx^{0}\right\} $: 
\begin{align}
\bx^{t+1} & =\bA_{s}\bm^{t}-\mathsf{b}_{t}\bm^{t-1},\\
\bm^{t} & =f_{t}\left(\bx^{t}\right),\\
\mathsf{b}_{t} & =\E\left[\frac{1}{N}{\rm div}f_{t}\left(\bZ^{t}\right)\right]
\end{align}
along with its state evolution $\left\{ \Kappa_{s,t}\middle\vert f_{t},\bx^{0}\right\} $
(see section \ref{subsec:Symmetric-SE} for a more complete definition
of these quantities). Note that assumptions \ref{assumption:gaussian_matrix}-\ref{assumption:gaussian-conv}
are satisfied because of \ref{assumption:gaussian-asym}-\ref{assumption:gaussian-conv-asym}.

Note that here $\Kappa_{2t,2t+1}=0$. It is also easy to identify
that 
\begin{align}
\bv^{t} & =\left(x_{1}^{2t+1},\dots,x_{m}^{2t+1}\right),\\
\bu^{t} & =\left(x_{m+1}^{2t},\dots,x_{N}^{2t}\right),\\
\Sigma_{s,t} & =\Kappa_{2s+1,2t+1},\\
\Tau_{s,t} & =\Kappa_{2s,2t}.
\end{align}
Applying Theorem \ref{thm:AMP_convergence} to the AMP recursion $\left\{ \bx^{t},\bm^{t}\vert f_{t},\bx^{0}\right\} $
shows our theorem. 
\end{proof}

\begin{proof}[Proof of Corollary \ref{coro:Empirical}]
The proof is by induction over $t$. Let $\cH_t$ be the claim that $\|\bu^s-\hbu^s\|_2/\sqrt{n} \approxP 0$ for all $s\le t$
and $\|\bv^s-\hbv^s\|_2/\sqrt{n} \approxP 0$ for all $s\le t-1$. The initial conditions imply immediately $\cH_0$.

We now prove that $\cH_t$ implies $\cH_{t+1}$.
Taking the difference of Eq.~(\ref{eq:FIRSTasymmetricAMP_2}) and Eq.~(\ref{eq:asymmetricAMP_2_emp}) and using triangular inequality,
we get
\begin{align}
\|\bv^{t}-\hbv^{t}\|_2&\le\|\bA\|_{\op} \|e_t(\bu^t)-e_t(\hbu^t)\|_2+|\sb_t-\hsb_t|\, \|g_{t-1}(\bv^{t-1})\|_2+
|\hsb_t|\, \|g_{t-1}(\hbv^{t-1})-g_{t-1}(\bv^{t-1})\|_2\\
& \le C_0(\delta) L \|\bu^t-\hbu^t\|_2 +|\sb_t-\hsb_t|\, \|g_{t-1}(\bv^{t-1})\|_2+ L |\hsb_t|\, \|\hbv^{t-1}-\bv^{t-1}\|_2\, ,
\end{align}
where  $L$ is the maximum Lipschitz constant of $e_t$ and $g_{t-1}$ and the second inequality holds with high probability
by the Bai-Yin law \cite{Bai1988}. Next notice that, with high probability,   $\|g_{t-1}(\bv^{t-1})\|_2/\sqrt{n}\le C$ for some constant $C$ by 
Theorem \ref{thm:AMP_convergence_asym} (together with Assumption \ref{assumption:gaussian-conv-asym}) and that 
$|\hsb_t|\le |\sb_t|+|\hsb_t-\sb_t|\le L+1$ with high probability by Assumption (\ref{eq:ConsistentOnsager}) and the Lipschitz continuity 
of $e_t$. Hence, for a suitable constant $C_1$, the following holds with high probability
\begin{align}
\frac{1}{\sqrt{n}}\|\bv^{t}-\hbv^{t}\|_2&\le C_1 \left\{\frac{1}{\sqrt{n}} \|\bu^t-\hbu^t\|_2 +|\sb_t-\hsb_t|+ \frac{1}{\sqrt{n}} \|\hbv^{t-1}-\bv^{t-1}\|_2\right\}\, .
\end{align}
We therefore have $\|\bv^{t}-\hbv^{t}\|_2/\sqrt{n}\approxP 0$ by Eq.~(\ref{eq:ConsistentOnsager}) and the induction hypothesis.

Taking the difference of Eq.~(\ref{eq:FIRSTasymmetricAMP_2}) and Eq.~(\ref{eq:asymmetricAMP_2_emp}), we get
\begin{align}
&\|\bu^{t+1}-\hbu^{t+1}\|_2\le \|\bA\|_{\op} \|g_t(\bv^t)-g(\hbv^t)\|_2+|\sd_t-\hsd_t|\|e_t(\bu^t)\|_2 +|\hsd_t|\|e_t(\bu^t)-e_t(\hbu^t)\|_2\\
&\le C_0(\delta) L \|\bv^t-\hbv^t\|_2+|\sd_t-\hsd_t|\|e_t(\bu^t)\|_2 +L |\hsd_t|\|\bu^t-\hbu^t\|_2\, ,
\end{align}
and the proof is completed by the same argument as above.
\end{proof}

\section{Application to general compressed sensing}
\label{sec:Application-CS}

In this section we discuss how the general theory of Section \ref{sec:MainResults} applies to the problem 
of reconstructing an unknown signal $\btheta_0\in\reals^n$ from noisy linear measurements given by
\begin{align}
\by=\bA\btheta_0+\bw\, .\label{eq:GeneralCS}
\end{align}
Here, $\bA\in\R^{m\times n}$ is the (known) sensing matrix, $\by\in\R^{m}$ is the measurement vector and
$\bw$ is a noise vector, independent of $\bA$. We know $\by$ and
$\bA$, and are required  to reconstruct $\btheta_0$. As before, it is understood that we are really given a sequence of
problems indexed by the dimensions $n$, with $m(n)/n\to \delta$.

If $m<n$, the problem becomes underdetermined. Reconstruction of $\btheta_0$ can be possible if we have
some prior information. 
The prior knowledge can be encoded in  a suitably chosen sequence of denoising function $\eta_{t}:\R^{n}\rightarrow\R^{n}$, $t\in\N$ \cite{donoho2013accurate}.
Given such a denoising function, we consider the following AMP algorithm:
\begin{align}
\hbtheta^{t+1} & =\eta_{t}\left(\hbtheta^{t}+\bA^{\sT}\br^{t}\right)\, ,\label{eq:GeneralCS_AMP-1}\\
\br^{t} & =\by-\bA\hbtheta^{t}+\hsb_{t}\br^{t-1}\, .\label{eq:GeneralCS_AMP-2}
\end{align}
where the initialization is given by $\hbtheta^{0}=0$ and $\eta_{-1}\left(\, \cdot\, \right)=0$. We 
assume the Onsager coefficient $\hsb_t$ to be a function of $\hbtheta^0,\dots,\hbtheta^t$, and $\br^0,\dots,\br^{t-1}$,
but we will discuss concrete choices below.

\subsection{General theory}

We make the following assumptions:
\begin{enumerate}[font={\bfseries},label={(C\arabic*)}]
\item \label{assumption:gaussian-CS} The sensing matrix $\bA$ is Gaussian
with i.i.d. entries, $(A_{ij})_{i\le m, j\le n}\sim \normal\left(0,1/m\right)$. 
\item For each $t$, the sequence (in $n$) of denoisers $\eta_{t}:\R^{n}\rightarrow\R^{n}$
is uniformly Lipschitz. 
\item \label{assumption:initial-CS}$\|\btheta_0\|_{2}/\sqrt{n}$ converges to a constant as $n\to\infty$.
\item The limit $\sigma_w= \lim_{n\to\infty}\|\bw\|_2/\sqrt{m} \in [0,\infty)$ exists.
\item \label{assumption:mixed-conv-CS} For any $t\in\N$ and any $\sigma\geq0$,
the following limit exists and is finite: 
\begin{align}
\lim_{n\to\infty}\frac{1}{n}\E\big[\big\< \btheta_0,\eta_{t}\left(\btheta_0+\bZ\right)\big\> \big]
\end{align}
where $Z\sim\normal\left(0,\sigma^{2}\id_{n}\right)$. 
\item \label{assumption:gaussian-conv-CS} For any $s,t\in\N$ and any $2\times2$
covariance matrix $\bSigma$, the following limit exists and is finite:
\begin{align}
\lim_{n\to\infty}\frac{1}{n}\E\left[\left\< \eta_{s}\left(\btheta_0+\bZ\right),\eta_{t}\left(\btheta_0+\bZ^{'}\right)\right\> \right]\, ,
\end{align}
where $\left(\bZ,\bZ^{'}\right)\sim\normal\left(0,\bSigma\otimes \id_{n}\right)$. 
\end{enumerate}
The technical assumptions \ref{assumption:mixed-conv-CS} and \ref{assumption:gaussian-conv-CS}
ensure the existence of the limits in the following \emph{state evolution} recursion:
\begin{align}
\tau_{0}^{2} & =\sigma_{w}^{2}+\lim_{n\to\infty}\frac{1}{\delta n}\left\Vert \btheta_0\right\Vert _{2}^{2},\\
\tau_{t+1}^{2} & =\sigma_{w}^{2}+\lim_{n\to\infty}\frac{1}{\delta n}\mathbb{E}\left[\left\Vert \eta_{t}\left(\btheta_0+\tau_{t}Z\right)-\btheta_0\right\Vert _{2}^{2}\right],
\label{eq:SE-CS}
\end{align}
where $\bZ\sim\normal\left(0,\id_{n}\right)$.

State evolution predicts the asymptotic behavior of the estimates $\hbtheta^{1},\hbtheta^{2},\dots$
in terms of an iterative denoising process.
\begin{thm}\label{thm:AMP-CS} 
Under assumptions \ref{assumption:gaussian-CS}-\ref{assumption:gaussian-conv-CS},
consider the recursion (\ref{eq:GeneralCS_AMP-1})-(\ref{eq:GeneralCS_AMP-2}).
Assume that $\hsb_t(\hbtheta^0,\br^0,\dots,\br^{t-1},\hbtheta^t)$ satisfies
\begin{align}
\hsb_t \approxP \sb_{t} \equiv \frac{1}{m}\E\left[\div\,\eta_{t-1}\left(\btheta_0+\tau_{t-1}\bZ\right)\right],\quad \bZ\sim\normal\left(0,\id_{n}\right) \label{eq:normal_asym_AMP-3}\, .
\end{align}
Further assume that the state evolution sequence satisfies $\tau_s >\sigma_w$ for all $s\le t$. 
Then, for any sequences $\phi_{n}:\left(\R^{m}\right)^{2}\to\R$, $n\geq1$, and $\psi_{n}:\left(\R^{n}\right)^{2}\rightarrow\R$, $n\geq1$, of
uniformly pseudo-Lipschitz functions of order $k$ 
\begin{align}
\phi_{n}(\br^{t},\bw) & \approxP\E\big[\phi_{n}\big(\bw+\sqrt{\tau_{t}^{2}-\sigma_{w}^{2}}\,\bZ,\bw\big)\big]\, ,\\
\psi_{n}\big(\hbtheta^{t}+\bA^{\sT}\br^{t},\btheta_0\big) & \approxP\E\big[\psi_{n}\big(\btheta_0+\tau_{t}\bZ',\btheta_0\big)\big]\, ,
\end{align}
where $\bZ\sim\normal\left(0,\id_{m}\right)$ and $\bZ'\sim\normal\left(0,\id_{n}\right)$. 
\end{thm}
\begin{proof}
This is a special case of the asymmetric AMP of Eqs.~(\ref{eq:FIRSTasymmetricAMP_1}), (\ref{eq:FIRSTasymmetricAMP_2}), with 
\begin{align}
\bu^{t+1} & =\btheta_0-\big(\bA^{\sT}\br^{t}+\hbtheta^{t}\big)\, ,\label{eq:thm_AMP_CS_proof_u}\\
\bv^{t} & =\bw-\br^{t}\, ,\\
e_{t}(\bu) & =\eta_{t-1}\big(\btheta_0-\bu\big)-\btheta_0\, ,\\
g_{t}(\bv) & =\bv-\bw\, ,
\end{align}
and the initialization $\bu^{0}=-\btheta_0$. Assumptions \ref{assumption:gaussian-asym}-\ref{assumption:gaussian-conv-asym}
are satisfied thanks to assumptions \ref{assumption:gaussian-CS}-\ref{assumption:gaussian-conv-CS}. The claim follows from Theorem \ref{thm:AMP_convergence_asym}
and Corollary \ref{coro:Empirical}.
\end{proof}

\begin{remark}\label{rmk:MSE}
A special case of common interest is $\psi_{n}(\bx,\by)=\| \eta_t(\bx)-\by\|_{2}^{2}/n$, for which Theorem \ref{thm:AMP-CS}
yields
\begin{align}
\frac{1}{n}\left\Vert \hbtheta^{t+1}-\btheta_0\right\Vert _{2}^{2}&\approxP \frac{1}{n} \mathbb{E}\big[\big\Vert \eta_{t}\big(\btheta_0+\tau_{t}\bZ'\big)-\btheta_0\big\Vert _{2}^{2}\big]\\
& = \delta(\tau_{t+1}^2-\sigma_w^2)\, .
\end{align}
\end{remark}

\begin{remark}
Two choices of the coefficient $\hsb_t$ that satisfy the assumption (\ref{eq:normal_asym_AMP-3}) are:
\begin{itemize}
\item The empirical mean
\begin{align}
\hsb_t = \frac{1}{m}\div\,\eta_{t-1}\big(\hbtheta^{t-1}+\bA^{\sT}\br^{t-1}\big)\, . \label{eq:OnsagerCS}
\end{align}
Using Theorem \ref{thm:AMP-CS}, this satisfies the assumptions by induction, provided $\bx\mapsto \frac{1}{m}\div\eta_{t}(\bx)$ is uniformly Lipschitz for each $t$.
\item If $\bx\mapsto \frac{1}{m}\div\, \eta_{t}(\bx)$  is not uniformly Lipschitz, a smoothed version of Eq.~(\ref{eq:OnsagerCS}) achieves the same goal,
namely
\begin{align}
\hsb_t = \frac{1}{m}\E\left[\div\,\eta_{t-1}\big(\hbtheta^{t-1}+\bA^{\sT}\br^{t-1}+\eps_n\bZ\big)\right] ,
\end{align}
where the expectation is with respect to $\bZ\sim\normal(0,\id_n)$, and $\eps_n$ is a deterministic sequence that converges to $0$ sufficiently slowly.
Adapting the arguments of Section \ref{sec:AMP_eps_to_0}, it is possible to show that this choice satisfies the assumption  (\ref{eq:normal_asym_AMP-3}).
\end{itemize} 
We also note that, even if $\bx\mapsto \frac{1}{m}\div\eta_{t}(\bx)$ is not uniformly Lipschitz, the choice (\ref{eq:OnsagerCS}) can still satisfy the assumption 
(\ref{eq:normal_asym_AMP-3}). For instance, if $\eta_t(\,\cdot\,)$ if the soft thresholding denoiser (a case studied in \cite{donoho2009message,Bayati2011}),
then $\bx\mapsto \frac{1}{m}\div\,\eta_{t}(\bx)$ is discontinuous but nevertheless a standard weak convergence argument implies 
Eq.~(\ref{eq:normal_asym_AMP-3}).
\end{remark}

\subsection{Denoising by convex projection}

An important feature of the theory developed in the previous section is that the denoiser $\eta_t$ 
can be fairly general, and not induced by an underlying optimization problem. Nevertheless, it is interesting to
specialize the theory developed so far to cases with special additional structure. 

One possible approach towards reconstruction from noisy measurements, cf. Eq.~(\ref{eq:GeneralCS}), assumes that
$\btheta_0$ belongs to a closed convex body $\cK\subseteq \reals^n$. The reconstruction method of choice 
solves the constrained least squares problem
\begin{align}
\mbox{minimize } & \;\;\|\by-\bA\btheta\|_2^2\, ,\label{eq:ConvexLS_1}\\
\mbox{subject to } & \;\; \btheta\in\cK\, . \label{eq:ConvexLS_2}
\end{align}
Denoting by $\sP_{\cK}$ the projection onto the set $\cK$ (which is a $1$- Lipschitz denoiser), the corresponding 
AMP algorithm reads
\begin{align}
\hbtheta^{t+1} & =\sP_{\cK}\left(\hbtheta^{t}+\bA^{\sT}\br^{t}\right)\, ,\label{eq:CVX_AMP-1}\\
\br^{t} & =\by-\bA\hbtheta^{t}+\hsb_{t}\br^{t-1}\, ,\label{eq:CVX_AMP-2}
\end{align}
where $\hbtheta^0=0$ and $\hsb_t$ is an estimator of $\sb_t = (1/m)\E\left[\div\sP_{\cK}(\btheta_0+\tau_t\bZ)\right]$. In many cases of interest, such estimator is simply given by
$\hsb_t = (1/m) \div\,\sP_{\cK}(\hbtheta^{t}+\bA^{\sT}\br^{t})$. It is possible to show that fixed points of this iteration are stationary points of the 
least squares problem  (\ref{eq:ConvexLS_1}), (\ref{eq:ConvexLS_2}). 

The constraint $\btheta\in\cK$ is effective if $\cK$ accurately captures the structure of the signal $\btheta_0$.
We denote by $\cC_{\cK}(\btheta_0)$ the tangent cone of $\cK$ at $\btheta_0$, i.e. the smallest convex cone 
containing $\cK-\theta_0$.  This can also be defined as
\begin{align}
\cC_{\cK}(\btheta_0) = \Big\{\bv\in\reals^n: \;\; \lim_{\eps\to 0+}\frac{1}{\eps}d(\btheta_0+\eps\,\bv,\cK) =0\Big\}\, ,
\end{align}
with $d(\bx,S) \equiv \inf\{\|\bx-\by\|_2:\by\in S\}$ the Euclidean point-set distance.
A highly structured signal $\btheta_0$ corresponds to a `small' cone  $\cC_{\cK}(\btheta_0)$.
This can be quantified via its statistical dimension \cite{chandrasekaran2012convex,amelunxen2014living}
\begin{align}
\Delta(\cC) = \E\big\{\big\|\sP_{\cC}(\bZ)\big\|_2^2\big\}\, ,
\end{align}
where expectation is with respect to $\bZ\sim\normal(0,\id_n)$.
It turns out that the statistical dimension also controls the convergence of AMP. As for our general theory, we will consider a sequence of problems indexed by 
the dimension $n$.
\begin{thm}\label{thm:ConvexAMP}
Consider the AMP iteration (\ref{eq:CVX_AMP-1}), (\ref{eq:CVX_AMP-2}), for  a sequence of problems $(\btheta_0(n), \bA(n), \cK(n),\bw(n))$ whereby $\bA= \bA(n)\in\reals^{m\times n}$ is a matrix with i.i.d. Gaussian entries 
$(A_{ij})_{i\le m, j\le n}\sim_{iid}\normal(0,1/m)$, $\cK=\cK(n)\subseteq \reals^n$ is a closed convex set with $\limsup_{n\to\infty}\max_{\bx \in\cK(n)}\|\bx\|_2/\sqrt{n}<\infty$,
$\btheta_0\in \cK(n)$, and $\lim_{n\to\infty}\|\bw(n)\|_2/\sqrt{m}=\sigma_w$. Assume $m/n\to \delta\in (0,\infty)$ and 
\begin{align}
\limsup_{n\to\infty} \frac{1}{m} \Delta\big(\cC_{\cK(n)}(\btheta_0(n))\big)  \le \rho\in [0,1)\, .
\end{align}
Then for any $t\ge 0$, letting $\sR_0 \equiv \limsup_{n\to\infty} \|\btheta_0(n)\|_2/\sqrt{n}$, we have
\begin{align}
\limsup_{n\to \infty}\frac{1}{n}\E\big\{\|\hbtheta^t-\btheta_0\|_2^2\big\}\le \delta\sR_0^2\, \rho^{t+1} + \delta\sigma_w^2 \frac{\rho-\rho^{t+1}}{1-\rho} \label{eq:SEConvex}\, .
\end{align}
\end{thm}
The proof of this statement is deferred to Appendix \ref{app:ProofConvex}. 

This theorem establishes exponentially fast convergence (in the high-dimensional limit) in all the region $m\ge (1+\eta)   \Delta_n$, $\Delta_n=\Delta \big(\cC_{\cK(n)}(\btheta_0(n))\big)$, i.e. 
whenever exact reconstruction is possible in absence of noise \cite{amelunxen2014living}.  Further, the convergence rate is precisely given by the ratio of the number of necessary measurements
to the number of measurements $\Delta_n/m$. 
For instance, it implies that, in order to achieve accuracy $\|\hbtheta^t-\btheta_0\|_2/\|\btheta_0\|_2\le \eps$ in the noiseless case $\sigma_w=0$, it is sufficient to run the AMP iteration
(\ref{eq:CVX_AMP-1}), (\ref{eq:CVX_AMP-2}) for approximately $\log(1/\eps)/\log(m/\Delta_n)$ iterations. 

The first result of this type (for separable soft-thresholding denoising) was obtained in \cite{donoho2009message,donoho2011noise}.
The only comparable result is obtained in recent work by Oymak, Recht, and Soltanolkotabi \cite{oymak2015sharp}, which establishes exponential convergence of 
of projected gradient descent, in a non-asymptotic sense, although at a slower rate\footnote{The same paper also prove convergence at a faster rate, but this requires $m>2\Delta_n$,
i.e. a number of measurements that is twice as large as the optimal one.}. In particular, in the noiseless case, $\eps$ accuracy requires $(n/m)\log(1/\eps)$. 
It would be interesting to derive a non-asymptotic version of Theorem \ref{thm:ConvexAMP}, which might be possible using the approach of \cite{rush2016finite}.

\section*{Acknowledgements}

This work was partially supported by grants NSF CCF-1319979, NSF DMS-1613091, NSF CCF-1714305.

\appendix

\section{Technical aspects of the numerical simulations}

\subsection{Matrix compressed sensing}
\label{app:Matrix}

Here we state the formula for computing the divergence of the singular value soft thresholding operator. Recall that for a matrix $\bY\in\reals^{n_1\times n_2}$, with singular value decomposition
\begin{align}
\bY=\sum_{i=1}^{n_{1}\wedge n_{2}}\sigma_{i}\bu_{i}\bv_{i}^{\sT},
\end{align}
the SVT operator with threshold $\lambda$ yields
\begin{equation}
\SVT(\bY;\lambda)=\sum_{i=1}^{n_{1}\wedge n_{2}}(\sigma_{i}-\lambda)_{+}\bu_{i}\bv_{i}^{\sT}.
\end{equation}
As proved in \cite{Candes}, the divergence for this operator can be computed using the formula
\begin{equation}
\div\,\SVT(\bY;\lambda) = \sum_{i=1}^{n_{1}\wedge n_{2}} \left[ \bfone_{\{\sigma_i > \lambda\}} + |m-n|\left(1-\frac{\lambda}{\sigma_i}\right)_+ \right] + 2 \sum_{i \neq j, i,j=1}^{n_1 \wedge n_2} \frac{\sigma_i(\sigma_i-\lambda)_+}{\sigma_i^2-\sigma_j^2}.
\end{equation}
This expression should be understood in a weak sense as it is not defined on the negligible set where $\bY$ has repeated singular values. 

\subsection{Compressed sensing with images}
\label{app:CSImages}

In our simulation, to compute the state evolution iterates 
\begin{align}
\tau_{0}^{2} & =\sigma_{w}^{2}+\lim_{n\to\infty}\frac{1}{\delta n}\left\Vert \btheta_0\right\Vert _{2}^{2},\\
\tau_{t+1}^{2} & =\sigma_{w}^{2}+\lim_{n\to\infty}\frac{1}{\delta n}\mathbb{E}\left[\left\Vert \eta_{t}\left(\btheta_0+\tau_{t}\bZ\right)-\btheta_0\right\Vert _{2}^{2}\right],
\end{align}
we approximated them by their non-asymptotic estimates: 
\begin{align}
\hat{\tau}_{0}^{2} & =\sigma_{w}^{2}+\frac{1}{\delta n}\left\Vert \btheta_0\right\Vert _{2}^{2},\\
\label{eq:tau_hat}
\hat{\tau}_{t+1}^{2} & =\sigma_{w}^{2}+\frac{1}{\delta n}\mathbb{E}\left[\left\Vert \eta_{t}\left(\btheta_0+\hat{\tau}_{t}\bZ\right)-\btheta_0\right\Vert _{2}^{2}\right].
\end{align}
Here $n = 170 \times 170$ is the size of our image. However, we could not compute the expectation in equation (\ref{eq:tau_hat}) exactly. Thus at each iteration we used a Monte Carlo method to approximate the expectation with the mean over 10 samples. Computing each sample amounts to adding gaussian noise of variance $\hat{\tau}_t^2$ over the Lena image, denoising with NLM, and computing the square error. The resulting state evolution is shown in figure \ref{fig:Lena_reconstruction_se}.

\section{Some useful tools}

We reminder the readers of three well-known results. The first concerns
with the operator norm of $\bA\in{\rm GOE}\left(n\right)$; see e.g.
\cite{Bai1988} for a more general statement. The second is a simple
consequence of Stein's lemma \cite{Stein1972}. The last one is the
Gaussian Poincar\'e inequality. 
\begin{thm}
\label{thm:A_norm} Consider a sequence of matrices $\bA\sim{\rm GOE}\left(n\right)$.
Then $\left\Vert \bA\right\Vert _{{\rm op}}\to2$ almost surely as $n\to\infty$. 
\end{thm}
\begin{lem}[Stein's lemma \cite{Stein1972}]
\label{lem:Stein}For any $2\times2$ covariance matrix $\Kappa$
and $\left(\bZ_{1},\bZ_{2}\right)\sim\normal\left(0,\Kappa\otimes \id_{n}\right)$,
and any $\varphi:\mathbb{R}^{n}\to\mathbb{R}^{n}$ such that $\frac{\partial\varphi_{i}}{\partial z_{i}}$
exists almost everywhere for $1\leq i\leq n$, if $\mathbb{E}\left[\left\< \bZ_{1},\varphi\left(\bZ_{2}\right)\right\> \right]$
and $\mathbb{E}\left[{\rm div}\varphi\left(\bZ_{2}\right)\right]$ exist,
then 
\begin{equation}
\mathbb{E}\left[\left\< \bZ_{1},\varphi\left(\bZ_{2}\right)\right\> \right]=\Kappa_{1,2}\mathbb{E}\left[{\rm div}\varphi\left(\bZ_{2}\right)\right]=\mathbb{E}\left[\frac{1}{n}\left\< \bZ_{1},\bZ_{2}\right\> \right]\mathbb{E}\left[{\rm div}\varphi\left(\bZ_{2}\right)\right].
\end{equation}
\end{lem}
\begin{thm}[Gaussian Poincar\textipa{é} inequality \cite{Boucheron2013}]
\label{thm:Poincare}Let $\bz\sim\normal\left(0,\id_{n}\right)$ and
$\varphi:\mathbb{R}^{n}\to\mathbb{R}$ continuous, weakly differentiable.
Then for some universal constant $c$, 
\begin{equation}
{\rm Var}\left[\varphi\left(\bz\right)\right]\leq c\,\mathbb{E}\left[\left\Vert \nabla\varphi\left(\bz\right)\right\Vert _{2}^{2}\right].
\end{equation}
\end{thm}
We state some properties of the GOE matrices, and provide proofs for
completeness. 
\begin{lem}
\label{lem:GOE_properties} Consider a sequence of matrices $\bA\sim{\rm GOE}\left(n\right)$
and two sequences (in $n$) of (non-random vectors) $\bu,\bv\in\mathbb{R}^{n}$
such that $\left\Vert \bu\right\Vert _{2}=\left\Vert \bv\right\Vert _{2}=\sqrt{n}$.
\begin{enumerate}[label=(\alph*)]
\item \label{item:asymptotic_orthogonality} $\frac{1}{n}\left\< \bv,\bA\bu\right\> \stackrel{{\rm P}}{\longrightarrow}0$ 
\item Let $\bP\in\mathbb{R}^{n\times n}$ be a sequence of projection matrices
such that there exists a constant $t$ that satisfies for all $n$,
${\rm rank}\left(\bP\right)\leq t$. Then $\frac{1}{n}\left\Vert \bP\bA\bu\right\Vert _{2}^{2}\xrightarrow[n\to\infty]{{\rm P}}0$. 
\item \label{item:norm_concentration} $\frac{1}{n}\left\Vert \bA\bu\right\Vert _{2}^{2}\stackrel{{\rm P}}{\longrightarrow}1$. 
\item Let $k$ be any positive integer. There exists a sequence (in $n$)
of random vectors $\bz\sim\normal\left(0,\id_{n}\right)$ such that for
any sequence $\varphi_{n}:\R^{n}\to\R$, $n\geq1$ of uniformly pseudo-Lipschitz
function of order $k$, 
\begin{equation}
\varphi_{n}\left(\bA\bu\right)\approxP\varphi_{n}\left(\bz\right).
\end{equation}
\end{enumerate}
\end{lem}
\begin{proof}
$\,$
\begin{enumerate}[label=\emph{(\alph*)}]
\item Recall that $\bA=\bG+\bG^{\sT}$ where $G_{i,j}$ are i.i.d. $\normal\left(0,1/(2n)\right)$
random variables, thus 
\begin{equation}
\frac{1}{n}\left\< \bv,\bA\bu\right\> =\frac{1}{n}\< \bv,\bG\bu\> +\frac{1}{n}\< \bv,\bG^{\sT}\bu\> .
\end{equation}
The random variable $\frac{1}{n}\left\< \bv,G\bu\right\> $
is centered Gaussian with variance 
\begin{equation}
\frac{1}{n^{2}}\sum_{i,j=1}^{n}v_{i}^{2}u_{j}^{2}\frac{1}{2n}=\frac{\Vert \bu\Vert_{2}^{2}\Vert \bv\Vert_{2}^{2}}{2n^{3}}=\frac{1}{2n}\longrightarrow0.
\end{equation}
Thus $\frac{1}{n}\< \bv,\bG\bu\> $ converges in probability
to 0. We can conclude as similarly, $\frac{1}{n}\< \bv,\bG^{\sT}\bu\> $
also converges in probability to 0. 
\item Consider $\bv_{1},\dots,\bv_{k}$ an orthogonal basis of the image of  $\bP$, such
that $\Vert \bv_{1}\Vert=\dots=\Vert \bv_{k}\Vert=\sqrt{n}$. Note that
$k$ can depend on $n$, but $k$ is uniformly bounded by $t$. Then, by point $(a)$,
\begin{equation}
\frac{1}{n}\Vert \bP\bA\bu\Vert_{2}^{2}=\frac{1}{n}\sum_{j=1}^{k}\left(\frac{\< \bA\bu,\bv_{j}\>}{\Vert \bv_{j}\Vert}\right)^{2}=\sum_{j=1}^{k}\left(\frac{1}{n}\< \bA\bu,\bv_{j}\>\right)^{2}\xrightarrow[n\to\infty]{}0
\end{equation}
using that $k\leq t$ for all $n$. 
\item This follows immediately from point $(d)$ below.
\item It is easy to check that $\bA\bu$ is a centered Gaussian vector with covariance matrix
$\bSigma=\id_{n}+\frac{1}{n}\bu\bu^{\sT}$. Thus there exists a Gaussian vector
$\bz\sim\normal\left(0,\id_{n}\right)$ such that $\bA\bu=\bSigma^{1/2}\bz=\bz+(\sqrt{2}-1)\frac{1}{n}\bu\bu^{\sT}\bz$.
Using that $\varphi$ is uniformly pseudo-Lipschitz of order $k$,
one has 
\begin{equation}
\left\vert \varphi(\bA\bu)-\varphi(\bz)\right\vert \leq L\left(1+\left(\frac{\left\Vert \bA\bu\right\Vert _{2}}{\sqrt{n}}\right)^{k-1}+\left(\frac{\left\Vert \bz\right\Vert _{2}}{\sqrt{n}}\right)^{k-1}\right)\frac{\left\Vert \bA\bu-\bz\right\Vert _{2}}{\sqrt{n}}.
\end{equation}
The law of large numbers gives $\left\Vert \bz\right\Vert_2 /\sqrt{n}\xrightarrow[n\to\infty]{}1$,
and we have $\left\Vert \bA\bu\right\Vert _{2}/\sqrt{n}\le \|\bSigma^{1/2}\|_{\rm op} \|\bz\|_2/\sqrt{n} \le \sqrt{2}  \|\bz\|_2/\sqrt{n} \xrightarrow[n\to\infty]{}\sqrt{2}$.
Further
\begin{equation}
\frac{\left\Vert \bA\bu-\bz\right\Vert _{2}}{\sqrt{n}}=\frac{\left\Vert \left(\bSigma^{1/2}-\id_{n}\right)\bz\right\Vert _{2}}{\sqrt{n}}=\frac{1}{n^{3/2}}(\sqrt{2}-1)\Vert \bu\bu^{\sT}\bz\Vert_{2}=(\sqrt{2}-1)\frac{1}{n}\vert \bu^{\sT}\bz\vert\xrightarrow[n\to\infty]{{\rm P}}0,
\end{equation}
where the last convergence follows from the fact that $\frac{1}{n}\bu^{\sT}\bz$
is a centered Gaussian random variable with variance $\Vert \bu\Vert_{2}^{2}/n^{2}=1/n$. 
\end{enumerate}
\end{proof}
We state some useful properties of uniformly pseudo-Lipschitz functions.
We omit the proofs, which are easy to verify. 
\begin{lem}
\label{lem:pseudoLipschitz_innerProd}Let $k$ be any positive integer.
Consider two sequences $f:\mathbb{R}^{n}\to\mathbb{R}^{n}$, $n\geq1$
and $g:\mathbb{R}^{n}\to\mathbb{R}^{n}$, $n\geq1$ of uniformly pseudo-Lipschitz
functions of order $k$. The sequence of functions $\varphi:\mathbb{R}^{n}\times\mathbb{R}^{n}\to\mathbb{R}$,
$n\geq1$, such that $\varphi\left(\bx,\by\right)=\left\< f\left(\bx\right),g\left(\by\right)\right\> $
is uniformly pseudo-Lipschitz of order $2k$. 
\end{lem}
\begin{lem}
\label{lem:pseudoLipschitz_add_constVec}Let $t$, $s$ and $k$ be
any three positive integers. Consider a sequence (in $n$) of $\bx_{1},\dots,\bx_{s}\in\mathbb{R}^{n}$
such that $\frac{1}{\sqrt{n}}\left\Vert \bx_{j}\right\Vert \leq c_{j}$
for some constant $c_{j}$ independent of $n$, for $j=1,\dots,s$,
and a sequence (in $n$) of order-$k$ uniformly pseudo-Lipschitz
functions $\varphi_n:\left(\mathbb{R}^{n}\right)^{t+s}\to\mathbb{R}$.
The sequence of functions $\phi_n\left(\,\cdot\,\right)=\varphi_n\left(\,\cdot\,,x_{1},\dots,x_{s}\right)$
is also uniformly pseudo-Lipschitz of order $k$. 
\end{lem}
\begin{lem}
\label{lem:pseudoLipschitz_expectation_wrt_Gaussian}Let $t$ be any
positive integer. Consider a sequence (in $n$) uniformly pseudo-Lipschitz
functions $\varphi_n:\left(\mathbb{R}^{n}\right)^{t}\to\mathbb{R}$
of order $k$. The sequence of functions $\phi_n:\left(\mathbb{R}^{n}\right)^{t}\to\mathbb{R}$
such that $\phi_n\left(\bx_{1},\dots,\bx_{t}\right)=\mathbb{E}\left[\varphi_n\left(\bx_{1},\dots,\bx_{t-1},\bx_{t}+\bZ\right)\right]$,
in which $\bZ\sim\normal\left(0,a\id_{n}\right)$ and $a\geq0$, is also
uniformly pseudo-Lipschitz of order $k$. 
\end{lem}
Finally, we have the following result on the Gaussian concentration
for uniformly pseudo-Lipschitz functions. 
\begin{lem}
\label{lem:psedoLipschitz_concentration}Let $\bZ\sim\normal\left(0,\id_{n}\right)$.
Let $k$ be any positive integer. For any sequence (in $n$) of uniformly
pseudo-Lipschitz functions $\varphi:\mathbb{R}^{n}\to\mathbb{R}$
of order $k$, $\varphi\left(\bZ\right)\approxP\mathbb{E}\left[\varphi\left(\bZ\right)\right]$. 
\end{lem}
\begin{proof}
This is a straightforward application of Theorem \ref{thm:Poincare}.
In particular, by the definition of uniformly pseudo-Lipschitz functions
of order $k$, 
\begin{equation}
\mathbb{E}\left[\left\Vert \nabla\varphi\left(\bZ\right)\right\Vert _{2}^{2}\right]\leq\frac{L^{2}}{n}\mathbb{E}\left[\left(1+\left(\frac{1}{\sqrt{n}}\left\Vert \bZ\right\Vert _{2}\right)^{k-1}\right)^{2}\right]\leq\frac{2L^{2}}{n}\left(1+\mathbb{E}\left[\left(\frac{1}{\sqrt{n}}\left\Vert \bZ\right\Vert _{2}\right)^{2(k-1)}\right]\right).
\end{equation}
Since $\bZ\sim\normal\left(0,\id_{n}\right)$, the right-hand side goes
to $0$ as $n\to\infty$. The claim is proven. 
\end{proof}

\section{Proof of Theorem \ref{thm:ConvexAMP}}
\label{app:ProofConvex}

By assumption 
\begin{align}
\sR_* \equiv 2\limsup_{n\to\infty}\max_{\bx \in\cK(n)} \frac{1}{\sqrt{n}} \, \|\bx\|_2<\infty\, .
\end{align}
Note for all $n\ge n_0$, $\|\hbtheta^t\|_2,\|\btheta_0\|_2\le \sR_*\sqrt{n}$ for all $t$.

Next fix $t\ge 0$ and denote by $B_t$ the right-hand side of Eq.~(\ref{eq:SEConvex}). Assume by contradiction that 
$\limsup_{n\to \infty}\E\{\|\hbtheta^t(n)-\btheta_0(n)\|_2^2\}/n= B_t+\eps>B_t$. We can then find a subsequence $\{n_1(\ell)\}_{\ell\ge 1}$ along which
$\lim_{\ell\to \infty}\E\{\|\hbtheta^t(n_1(\ell))-\btheta_0(n_1(\ell))\|_2^2\}/n_1(\ell)= B_t+\eps$. We will prove that this subsequence can be further refined
to $\{n_2(\ell)\}_{\ell\ge 1}\subseteq \{n_1(\ell)\}_{\ell\ge 1}$  such that $\lim_{\ell\to \infty}\E\{\|\hbtheta^t(n_2(\ell))-\btheta_0(n_2(\ell))\|_2^2\}/n_2(\ell)\le B_t$,
thus leading to a contradiction.

To simplify the notation we can assume, without loss of generality, that the first subsequence is not needed, i.e. 
$\limsup_{n\to \infty}\E\{\|\hbtheta^t(n)-\btheta_0(n)\|_2^2\}/n= B_t+\eps>B_t$.
We then claim that we can find a subsequence $\{n_2(\ell)\}_{\ell\ge 1}$ along which Assumptions \ref{assumption:initial-CS}, \ref{assumption:mixed-conv-CS} and  
\ref{assumption:gaussian-conv-CS} hold,
with $\eta_s(\,\cdot\, )$, $\eta_t(\,\cdot\,) =\sP_{\cK}(\,\cdot\,)$.
Consider Assumption \ref{assumption:gaussian-conv-CS}.
Let the functions $F_n:S^2_+\to\reals$ (with $S^2_+$ the cone of $2\times 2$ positive
semidefinite matrices)  be defined by
\begin{align}
F_n(\bSigma) \equiv\frac{1}{n}\E\Big[ \big\<\sP_{\cK}(\btheta_0+\bZ),\sP_{\cK}(\btheta_0+\bZ')\big\>\Big]\, ,
\end{align}
where expectation is with respect to $(\bZ,\bZ')\sim\normal(0,\bSigma\otimes\id_n)$. 

Note that the function
$(\bZ,\bZ')\mapsto \big\<\sP_{\cK}(\btheta_0+\bZ),\sP_{\cK}(\btheta_0+\bZ')\big\>/n$ is uniformly pseudo-Lipschitz of order $2$. Hence, using Lemma 
\ref{lem:pseudo-lip_and_gaussians}, we have
\begin{align}
\sup_{n\ge 1}\big|F_n(\bSigma_1) - F_n(\bSigma_2)\big|\le  \xi(\bSigma_1,\bSigma_2)\, ,
\end{align}
for some function $\xi$ such that $\lim_{\bSigma_1\to\bSigma_2} \xi(\bSigma_1,\bSigma_2)=0$.
Further $\sup_{n\ge 1}|F_n(\bSigma)|\le \sR_*^2$. Hence by the Arzel\`a-Ascoli theorem, $F_n$ converges uniformly on any compact set
$\{\bSigma:\, \|\bSigma\|_F\le C\}$, thus satisfying condition \ref{assumption:gaussian-conv-CS}, along a certain subsequence $\{n'_2(\ell)\}_{\ell\ge 1}$.
Assumption  \ref{assumption:mixed-conv-CS} is established by the same argument, eventually refining the subsequence to $\{n''_2(\ell)\}_{\ell\ge 1}$.
Finally, by taking a further subsequence $\{n_2(\ell)\}_{\ell\ge 1}$, we can assume that $\|\btheta_0(n_2(\ell))\|_2^2/\sqrt{n}\to \sR_0$.

We can therefore apply Theorem \ref{thm:AMP-CS}  (and Remark \ref{rmk:MSE}) along this subsequence, to obtain 
$\|\hbtheta^{t+1}-\btheta_0\| _{2}^{2}/n\approxP\delta(\tau_{t+1}^2-\sigma_w^2)$ and hence (since $\|\hbtheta^{t+1}-\btheta_0\| _{2}^{2}/n\le\sR_*^2$
is bounded uniformly)
\begin{align}
\lim_{\ell\to\infty}\frac{1}{n}\E\big\{\|\hbtheta^{t+1}(n_2(\ell))-\btheta_0(n_2(\ell))\| _{2}^{2}\big\}= \delta(\tau_{t+1}^2-\sigma_w^2)\, . \label{eq:LimitMSESubseq}
\end{align}
Here $\tau_{t+1}$ is given recursively by Eq.~(\ref{eq:SE-CS}), namely $\tau^2_0 = \sR_0^2$ and
\begin{align}
\tau_{s+1}^{2} & =\sigma_{w}^{2}+G(\tau^2_s)\,,\\
G(\tau^2) & =  \lim_{\substack{\ell\to\infty \\ n=n_2(\ell)}}\frac{1}{n\delta}\E\left[\left\|\sP_{\cK}(\btheta_0+\tau\bZ)-\btheta_0\right\|_{2}^{2}\right]\,,
\end{align}
where the limit exists by the existence of the limit of $F_n(\bSigma)$ above. 
Now, since $\cK-\btheta_0\subseteq \cC_{\cK}(\btheta_0)$, we have
\begin{align}
\left\|\sP_{\cK}(\btheta_0+\tau\bZ)-\btheta_0\right\|_{2}^{2} = 
\left\|\sP_{\cC_{\cK}(\btheta_0)}\Big[\sP_{\cK}(\btheta_0+\tau\bZ)-\btheta_0\Big]\right\|_{2}^{2} \le
\big\|\sP_{\cC_{\cK}(\btheta_0)} (\tau\bZ)\big\|_2^2\, .
\end{align}
Therefore 
\begin{align}
G(\tau^2) & \le \limsup_{n\to\infty}\frac{1}{m}\E\big\{\big\|\sP_{\cC_{\cK}(\btheta_0)} (\bZ)\big\|_2^2\big\}\, \tau^2 \le \rho\, \tau^2\, .
\end{align}
We therefore get the recursion $\tau_{s+1}^{2} \le \sigma_{w}^{2}+ \rho \tau^2_s$, which can be summed to yield
\begin{align}
\tau_t^2 = \sR_0^2\rho^t+\sigma_w^2 \frac{1-\rho^t}{1-\rho}\, ,
\end{align}
Therefore, using Eq.~(\ref{eq:LimitMSESubseq}), we get
\begin{align}
\lim_{\substack{\ell\to\infty \\ n=n_2(\ell)}}\frac{1}{n}\E\big\{\|\hbtheta^{t+1}(n)-\btheta_0(n)\| _{2}^{2}\big\}\le B_t \, ,
\end{align}
which yields the desired contradiction hence proving the theorem.

 \bibliographystyle{amsalpha}
\newcommand{\etalchar}[1]{$^{#1}$}
\providecommand{\bysame}{\leavevmode\hbox to3em{\hrulefill}\thinspace}
\providecommand{\MR}{\relax\ifhmode\unskip\space\fi MR }
\providecommand{\MRhref}[2]{%
  \href{http://www.ams.org/mathscinet-getitem?mr=#1}{#2}
}
\providecommand{\href}[2]{#2}

\end{document}